\documentclass[12pt,reqno]{article}
\usepackage[english]{babel}
\usepackage{amsmath,amsthm}
\usepackage{amsfonts}
\usepackage{graphicx}
\usepackage{enumerate}
\usepackage[usenames,dvipsnames]{color}
\usepackage{subfigure}
\usepackage[right,pagewise,displaymath, mathlines]{lineno}
\usepackage{epstopdf}
\usepackage{color}
\usepackage{multirow}

\usepackage{tikz}
\usepackage{schemabloc}

\numberwithin{equation}{section}
\numberwithin{figure}{section}
\numberwithin{table}{section}

\newtheorem{theorem}{Theorem}[section]

\newtheorem{definition}{Definition}[section]

\newtheorem{note}{Note}[section]

\numberwithin{equation}{section}

\begin{document}

\begin{center}

{\Large\bf Detecting intrusions in control systems: a rule of thumb, its justification and illustrations}

\vspace*{7mm}

{\large Nadezhda Gribkova}

\medskip

\textit{Faculty of Mathematics and Mechanics, St.\,Petersburg State University,\\ St.\,Petersburg 199034, Russia}

\bigskip

{\large Ri\v cardas Zitikis}

\medskip

\textit{School of Mathematical and Statistical Sciences,
Western University, \break London, Ontario N6A 5B7, Canada}

\end{center}

\bigskip

\begin{quote}

{\bf Abstract.} Control systems are exposed to unintentional errors, deliberate intrusions, false data injection attacks, and various other disruptions. In this paper we propose, justify, and illustrate a rule of thumb for detecting, or confirming the absence of, such disruptions. To facilitate the use of the rule, we rigorously discuss background results that delineate the boundaries of the rule's applicability. We also discuss ways to further widen the applicability of the proposed intrusion-detection methodology.

\medskip


\medskip

%

\end{quote}


\section{Introduction}
\label{intro}

Computer systems monitor and control a myriad of physical processes, and their protection against random errors, deliberate intrusions (e.g., Denning, 1987; Debar et al., 1999; C\'{a}rdenas et al., 2011; Premathilaka et al., 2013; and references therein), false data injections (e.g., Liang, 2017; and references therein), and other disruptors has become of much interest. A number of sophisticated methods have been suggested in the literature for tackling such problems, including probabilistic (e.g., Huang et al., 2016; Onoda,~2016), deep learning (e.g., He et al.,~2017), and artificial neural networks based methods (e.g., Potluri,~2017), to name a few. The aim of the present paper is to describe, justify, and illustrate a simple-to-formulate and quick-to-implement procedure for detecting intrusions in control systems. The procedure is a natural offspring of extensive probabilistic and statistical explorations by Davydov and Zitikis (2007, 2017), Chen at al. (2018), and Gribkova and Zitikis (2018). Details follow.

Suppose we are dealing with a control system, which intakes random variables $X_{1},\dots , X_{n}$ and outputs their transformations $Y^0_{1},\dots , Y^0_{n}$; we use the superscript ``0'' to indicate that the outputs have not been compromised.
\begin{figure}[t!] \bigskip
\centering
\tikzstyle{line} = [draw, -stealth, thick]
\tikzstyle{block} = [draw, rectangle, text width=11em, text centered, minimum height=15mm, node distance=10em]
\begin{tikzpicture}
\node [block] (filter) {Transfer function $h(x)$};
\node [left of=filter, xshift=-14em] (inputs) {};
\node [right of=filter, xshift=14em] (outputs) {};
\path [line] (inputs) -- node[yshift=1em, xshift=-1em] {$\dots , X_{1},\dots , X_{n}\sim F$} (filter);
\path [line] (filter) -- node[yshift=1em, xshift=0em] {$\dots , Y_{1}^0,\dots , Y_{n}^0$} (outputs);
\end{tikzpicture}
\caption{A control system.}
	\label{fig-1}
\end{figure}
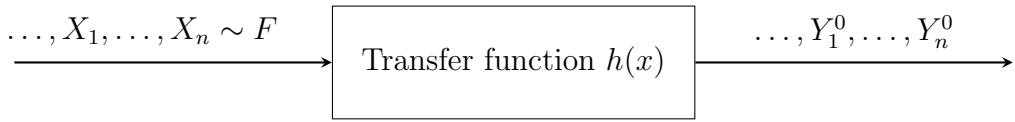
There is a transfer function $h(x)$ associated with this control system (Figure~\ref{fig-1}), and thus  the outputs are of the form
\begin{equation}\label{eq-0}
Y_i^0:=h(X_i), \quad i=1,\dots , n.
\end{equation}
Assume that the inputs $X_{1},\dots , X_{n}$ have been pre-whitened, and thus are independent and identically distributed (iid) random variables, whose marginal cumulative distribution functions (cdf's) we denote by $F(x)$.

Filters have pre-specified transfer windows, which are usually intervals $(a,b]$ for some real numbers $a<b$. Hence, the associated transfer function $h(x)$ maps the interval $(a,b]$ to the set of real numbers. We assume that the cdf $F(x)$ is supported by $(a,b]$, that is, $F(x)\in (0,1)$ for all $x\in (a,b)$, with the boundary values $F(a)=0$ and $F(b)=1$. Furthermore, we assume that the cdf $F(x)$ is strictly increasing on $(a,b]$, and thus generates data that can potentially fill in every part of the transfer window. Even though the methodology developed in this paper allows for various transfer functions, to facilitate clarity of the following arguments, we work with functions $h(x)$ that have continuous and bounded first derivatives.

The outputs $Y^0_{1},\dots , Y^0_{n}$ may, however, be compromised by intrusion variables, which we denote by $\varepsilon_{1},\dots , \varepsilon_{n}\sim F_{\varepsilon}$. We assume that they are iid random variables, independent of the inputs $X_{1},\dots , X_{n}$, and have means $\mu_{\varepsilon}=0 $ and finite variances $\sigma_{\varepsilon}^2<\infty $. These are, of course, the classical assumptions associated with errors in dynamical systems such as time series (e.g., Tong, 1990; Box et al., 2015).
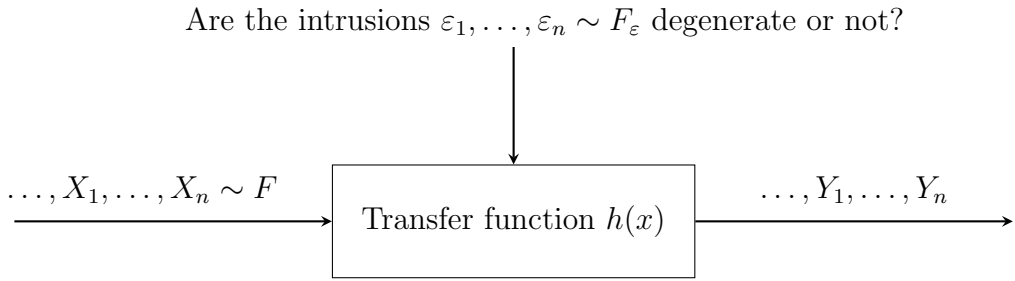
\begin{figure}[t!] \bigskip
\centering
\tikzstyle{line} = [draw, -stealth, thick]
\tikzstyle{block} = [draw, rectangle, text width=11em, text centered, minimum height=15mm, node distance=10em]
\begin{tikzpicture}
\node [block] (filter) {Transfer function $h(x)$};
\node [left of=filter, xshift=-14em] (inputs) {};
\node [right of=filter, xshift=14em] (outputs) {};
\node [above of=filter, yshift=4em] (errors) {\qquad Are the intrusions $\varepsilon_{1},\dots , \varepsilon_{n}\sim F_{\varepsilon} $ degenerate or not?};
\path [line] (errors) -- (filter);
\path [line] (inputs) -- node[yshift=1em, xshift=-1em] {$\dots , X_{1},\dots , X_{n}\sim F$} (filter);
\path [line] (filter) -- node[yshift=1em, xshift=0em] {$\dots , Y_{1},\dots , Y_{n}$} (outputs);
\end{tikzpicture}
\caption{The control system with a potential intrusion.}
	\label{fig-2}
\end{figure}
That is, the actual outputs are realizations of the random variables (Figure~\ref{fig-2})
\begin{equation}\label{eq-1}
Y_i:=Y_i^0+\varepsilon_i, \quad i=1,\dots , n.
\end{equation}
To find out whether or not the system is being compromised is the same -- in the context of the present paper -- as testing whether or not the cdf $F_{\varepsilon}$ is non-degenerate (i.e., $\varepsilon_i$'s are not equal to $0$) or degenerate (i.e., $\varepsilon_i$'s are equal to $0$). In the former case, we say that the intrusion variables are present, and in the latter case, we say that they are absent. Hence, in practical uses of the herein developed recommendations, the presence and the absence of intrusion variables should be  understood in the statistical sense.

Intuitively (e.g., C\'{a}rdenas et al., 2011, p.~360), in order to distinguish between the two situations, the outputs $Y^0_{1},\dots , Y^0_{n}$ under no intrusion should be in some  reasonable order so that incoming intrusion variables would disrupt the order and in this way make them detectible. The following definition clarifies what we mean by ``reasonable order'' within the context of the present research.

\begin{definition}\label{order}\rm
We say that the non-contaminated outcomes $Y_i^0=h(X_i)$, $ i=1,\dots , n$, are  \textit{in reasonable order} when the random sequence $B_n^0$, $n\ge 1$, defined by 
\begin{equation}
\label{condition-0}
B_n^0:={1\over \sqrt{n}} \sum_{i=2}^n \big | h(X_{i:n})-h(X_{i-1:n}) \big |
\end{equation}
is asymptotically bounded in probability, that is, when $B_n^0=O_{\mathbf{P}}(1)$ as $n\to \infty $, where $X_{1:n}\le \cdots \le X_{n:n}$ are the ordered inputs $X_1,\dots , X_n$. If, however, $B_n^0=O_{\mathbf{P}}(1)$ does not hold (e.g., when $B_n^0\stackrel{\mathbf{P}}{\to} \infty $), then we call the outputs \textit{out of reasonable order}.
\end{definition}

To work out intuition on this definition, consider first the case when the outputs are very chaotic in the sense that they are zigzagging up and down in a steady fashion when going from one order statistic $X_{i-1:n}$  to the next one $X_{i:n}$. Such outputs, obviously, are out of reasonable order in the sense of the above definition.

If, on the other hand, the function $h\circ F^{-1}(u)$ is smooth and well-behaving in the sense that its derivative $(\mathrm{d}/\mathrm{d}u)h\circ F^{-1}(u)$ exists and is bounded on its domain of definition $[0,1]$, then the outputs are in reasonable order as the following computations show:
\begin{align}
B_n^0&={1\over \sqrt{n}} \sum_{i=2}^n \big | h(X_{i:n})-h(X_{i-1:n}) \big |
\notag \\
&={1\over \sqrt{n}}\sum_{i=2}^n \big | h\circ F^{-1}(U_{i:n})-h\circ F^{-1}(U_{i-1:n})\big |
\notag \\
&={1\over \sqrt{n}}\sum_{i=2}^n \bigg | {\mathrm{d}\over \mathrm{d}t} h\circ F^{-1}(\nu_i)\big (U_{i:n}-U_{i-1:n}\big )\bigg |
\notag \\
&= O_{\mathbf{P}}(n^{-1/2}),
\label{condition-0b}
\end{align}
where $U_{i:n}$, $i=1,\dots , n$, are the order statistics of the uniform on $[0,1]$ random variables  $U_i:=F(X_i)$, $i=1,\dots , n$, and $\nu_i$ is a random variable taking values between $U_{i-1:n}$ and $U_{i:n} $. Note that  the assumed boundedness of the derivative $(\mathrm{d}/\mathrm{d}u)h\circ F^{-1}(u)$ results in the much stronger asymptotic property $B_n^0= O_{\mathbf{P}}(n^{-1/2})$ than the required $B_n^0= O_{\mathbf{P}}(1)$. This suggest that the assumption can be relaxed, which is indeed possible as we shall soon see.

\begin{note}\rm
In what follows, if a random sequence $\xi_n$ is such that $\xi_n=O_{\mathbf{P}}(1)$, then we simply say that $\xi_n$ is asymptotically bounded, without specifying ``in probability.'' If $\xi_n\stackrel{\mathbf{P}}{\to} c$, then we say that $\xi_n$ converges to $c$. Finally, if $\xi_n\stackrel{\mathbf{P}}{\to} \infty $, we say that $\xi_n$ tends to infinity.
\end{note}

We have organized the rest of the paper as follows. In Section~\ref{rule}, we formulate, justify, and illustrate a rule of thumb for detecting unintentional (e.g., measurement) errors as well as deliberate intrusions (e.g., false data injections) in control systems. In Section~\ref{determ}, we explore the rule when, due to a variety of reasons such as vulnerability testing, the system is fed artificially-deigned deterministic inputs, instead of the usual random inputs $X_1,\dots , X_n$. Section~\ref{conclude} concludes the paper with a brief summary of main contributions.

\section{A rule of thumb and its justification}
\label{rule}

The rule of thumb that we shall formulate in a moment relies on the asymptotic behaviour of the quantities 
\[
I_n:={1\over B_n \sqrt{n}} \sum_{i=2}^n \big ( Y_{i,n}-Y_{i-1,n} \big )_{+} 
\]
and 
\[
B_n:={1\over \sqrt{n}} \sum_{i=2}^n \big | Y_{i,n}-Y_{i-1,n} \big |  
\]
when $n\to \infty $, where $x_{+}=\max\{x,0\}$ for any $x\in \mathbf{R}$, and $Y_{1,n}, \dots, Y_{n,n}$ are the concomitants (e.g., David and Nagaraja,~2003) that arise from the pairs $(X_1,Y_1), \dots, (X_n,Y_n)$.

\begin{quote}{\bf Rule of Thumb} 
\begin{description} 
\item[Case 1:] $I_n$ does \textit{not} (decisively or vaguely) approach $1/2$:
\begin{enumerate}[(i)]
\item 
If $I_n$ decisively tends to a limit other than $1/2$, then the decision maker is to be advised about the \textit{absence} of intrusion variables.
\item 
If $I_n$ seems to tend to a limit other than $1/2$ but there is some doubt (due to, e.g., data uncertainty) as to whether this is indeed true, then we check if $B_n$ is asymptotically bounded, and if yes, then the decision maker is to be advised about the \textit{absence} of intrusion variables.
\end{enumerate}
\item[Case 2:]  $I_n$ (decisively or vaguely) approaches $1/2$: 
\begin{enumerate}[(i)]    
\item 
    If $B_n$ decisively grows to infinity when $n\to \infty $, then the decision maker is to be advised about the \textit{presence} of intrusion variables.
\item
    If $B_n$ is definitely asymptotically bounded, then the decision maker is to be advised about the \textit{absence} of intrusion variables.
\item 
If there is some doubt as to whether $B_n$ grows to infinity or is asymptotically bounded, then it should be checked whether or not the transfer function $h(x)$ takes the same value at the two endpoints of the transfer window, that is, whether or nor the equation $h(a)=h(b)$ holds (a justification provided after Theorem~\ref{Thm-C} below):
      \begin{itemize}
        \item If the equation does \textit{not} hold, then the decision maker is to be advised about the \textit{presence} of intrusion variables. 
        \item If the equation holds, then the usual work of the control system should be interrupted and deterministic inputs (details in Section~\ref{determ} below) fed into the system for running the entire rule of thumb anew from the very beginning.
      \end{itemize}
\end{enumerate}
\end{description}
\end{quote}

Figure~\ref{fig-pure}
\begin{figure}[t!]
\centering
\subfigure[$I_n$ when $\alpha=2$, $\beta=3$.]{%
\includegraphics[height=0.3\textwidth]{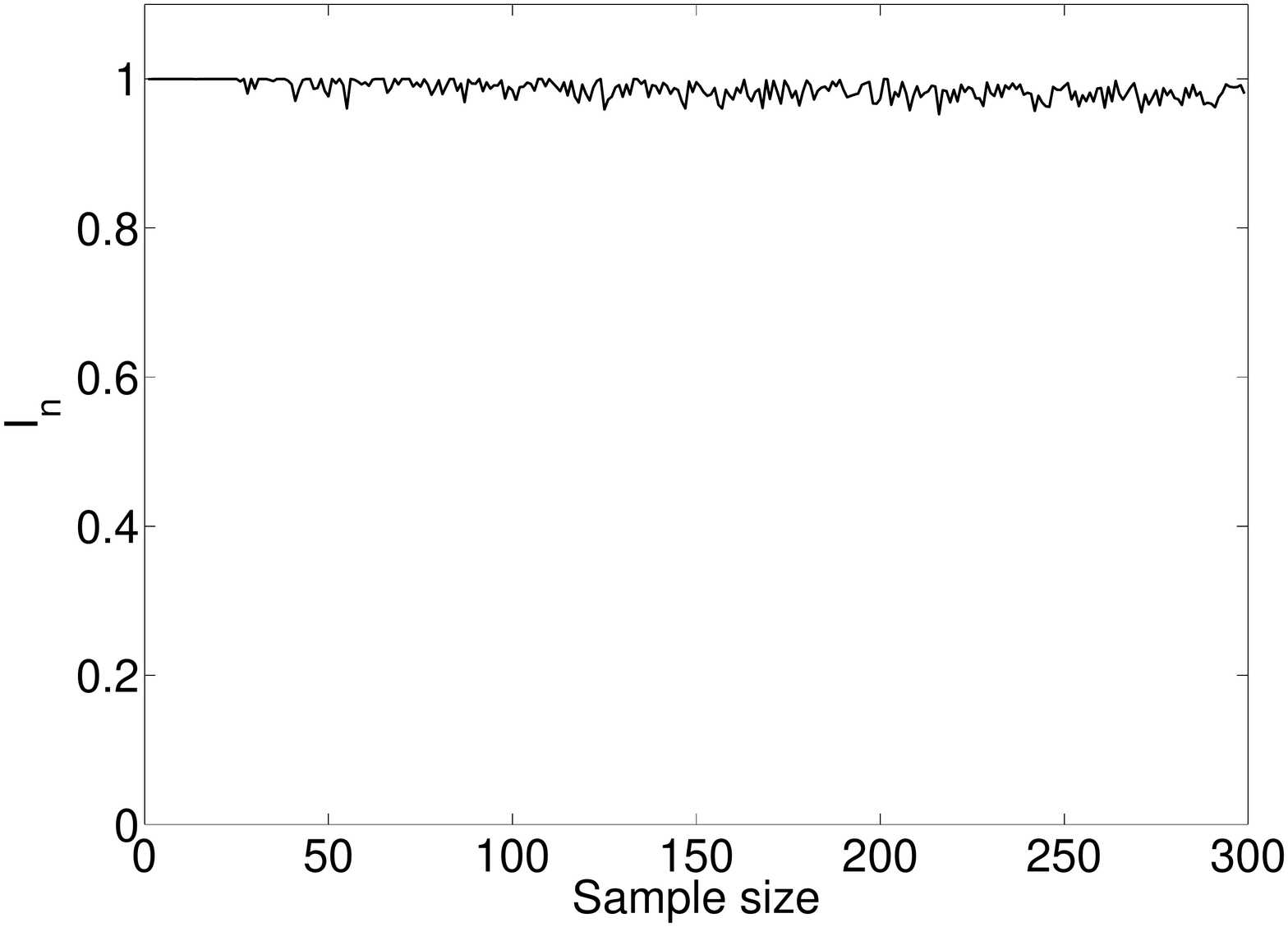}}
\hfill
\subfigure[$B_n$ when $\alpha=2$, $\beta=3$.]{%
\includegraphics[height=0.3\textwidth]{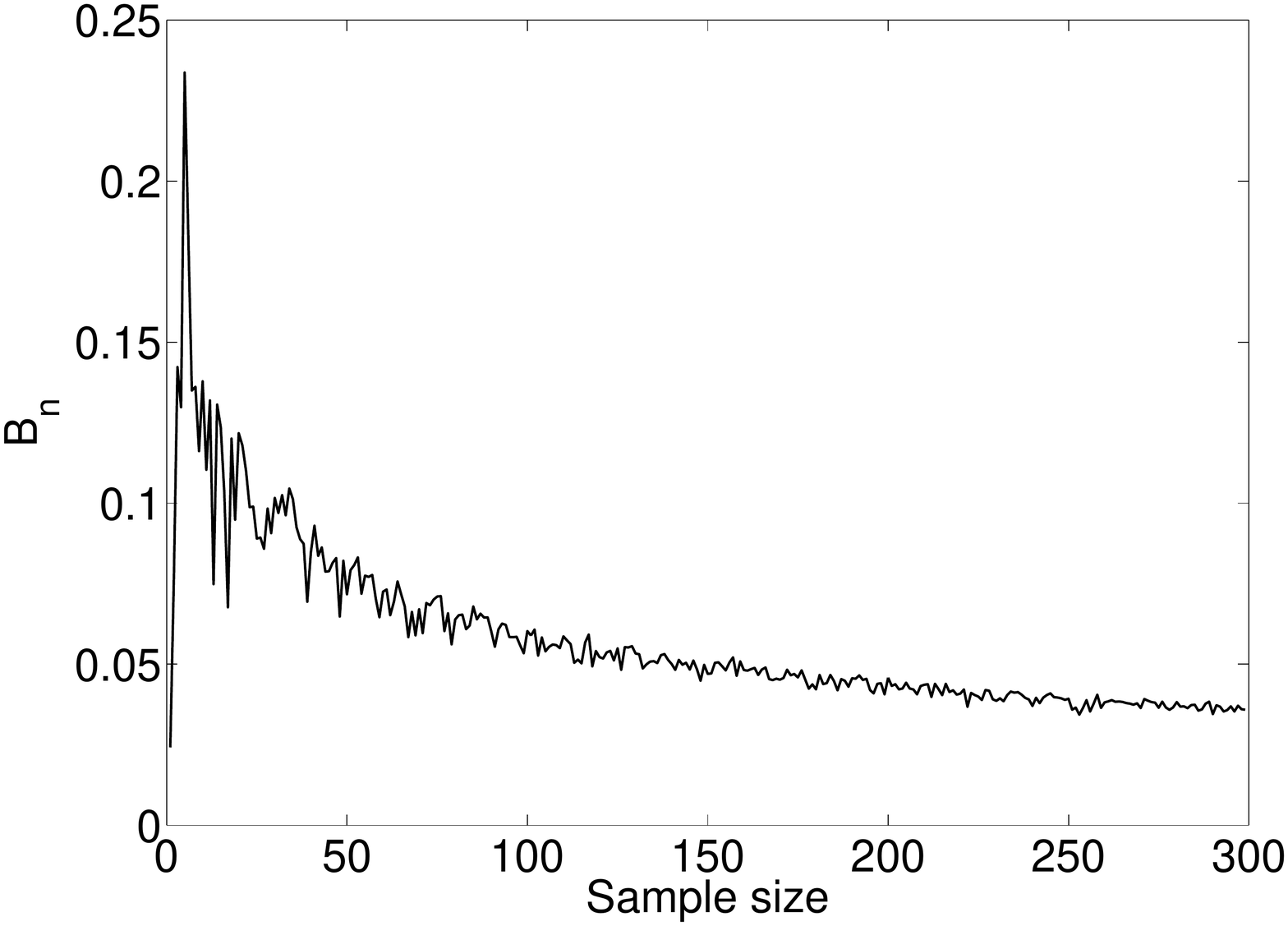}}
\\
\subfigure[$I_n$ when $\alpha=2$, $\beta=2$.]{%
\includegraphics[height=0.3\textwidth]{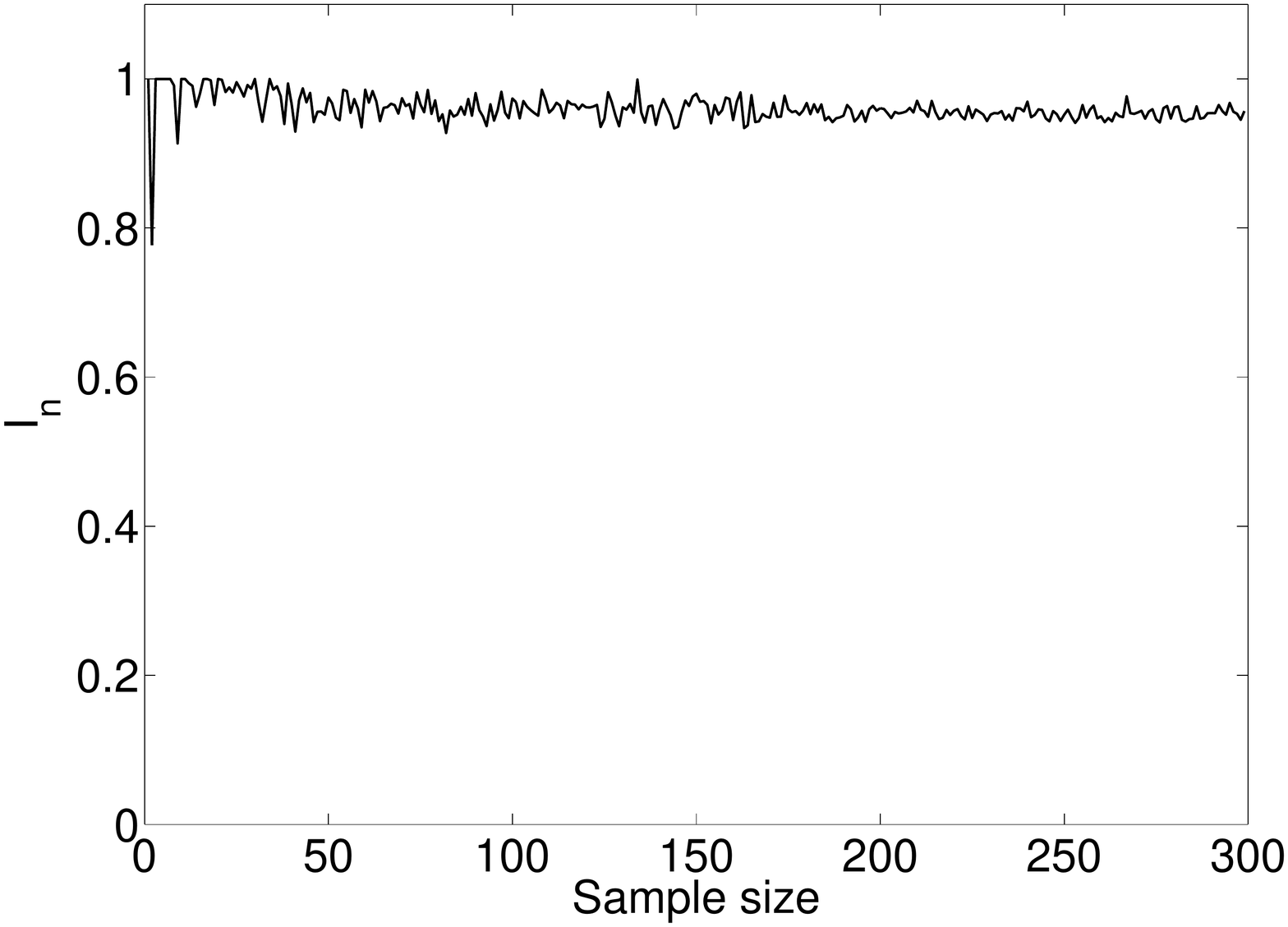}}
\hfill
\subfigure[$B_n$ when $\alpha=2$, $\beta=2$.]{%
\includegraphics[height=0.3\textwidth]{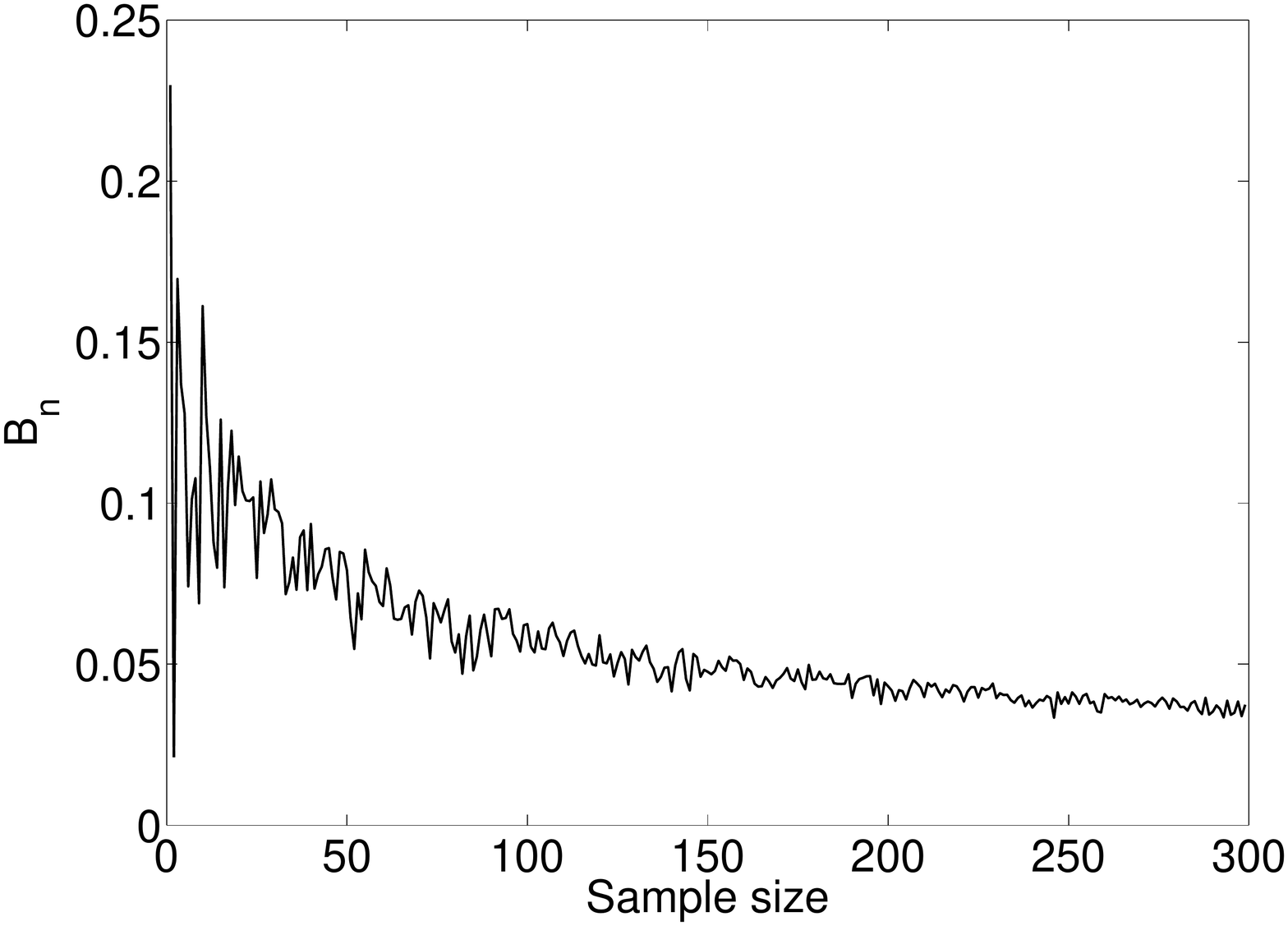}}
\\
\subfigure[$I_n$ when $\alpha=3$, $\beta=2$.]{%
\includegraphics[height=0.3\textwidth]{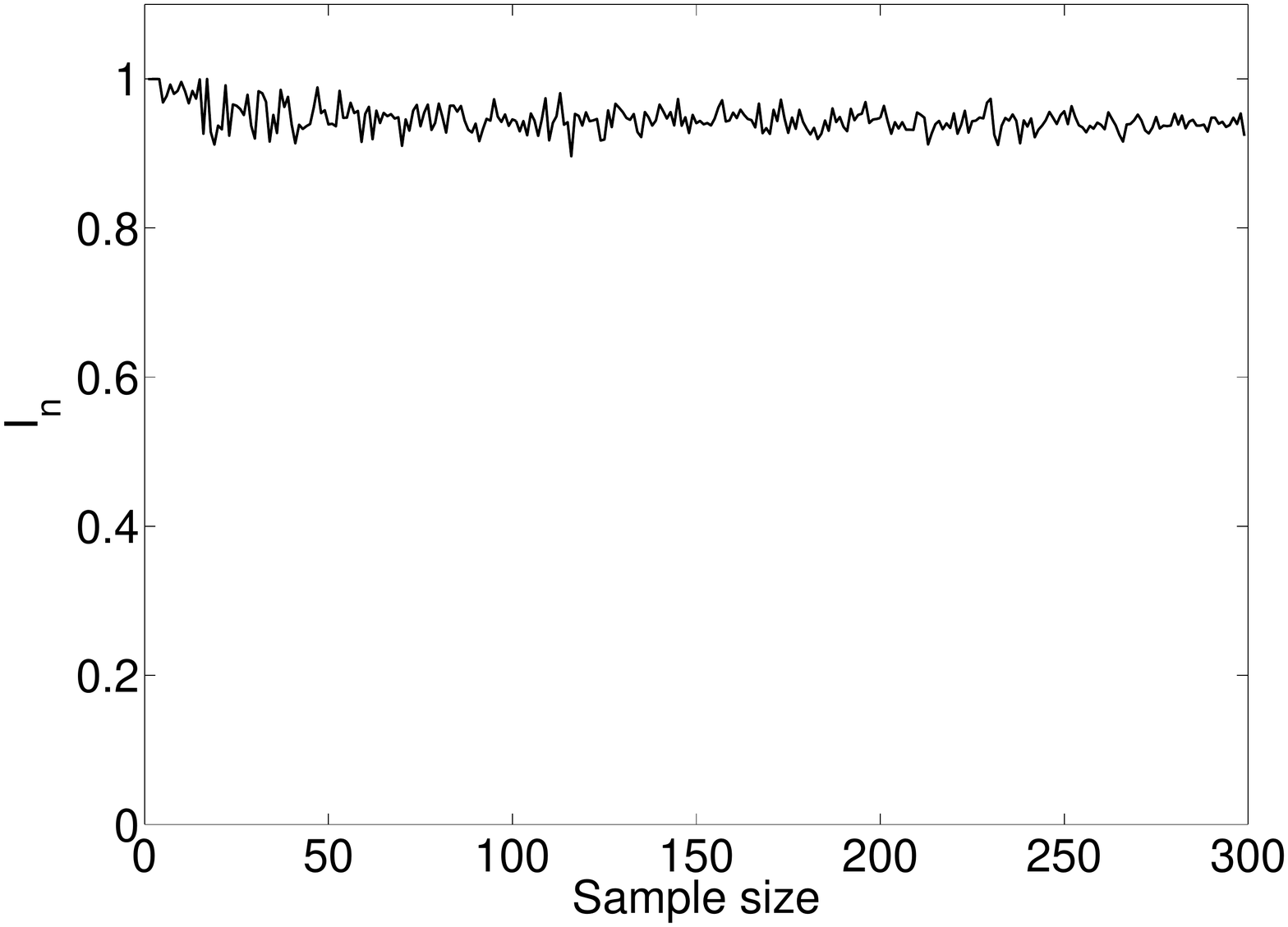}}
\hfill
\subfigure[$B_n$ when $\alpha=3$, $\beta=2$.]{%
\includegraphics[height=0.3\textwidth]{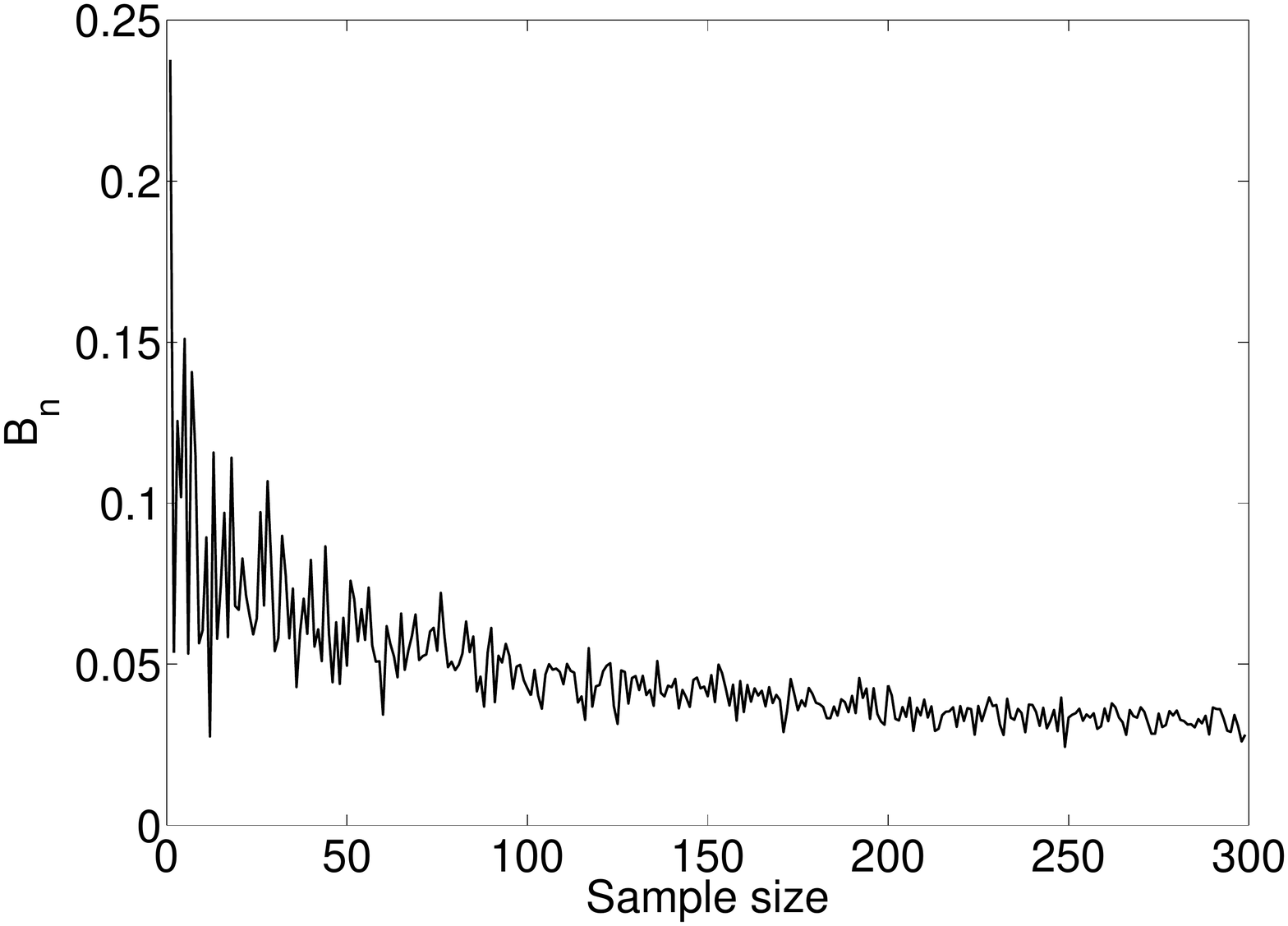}}
\caption{$I_n$ and $B_n$ with respect to the sample size $1\le n \le 300$ when $X_i\sim \mathrm{Beta}(\alpha,\beta)$,  $h(x)=1-(x-4/5)^2$ on $[0,1]$, and non-compromised outputs.}
\label{fig-pure}
\end{figure}
illustrates the rule of thumb when the outputs are not compromised by intrusion variables. Clearly, $I_n$ stays away from $1/2$, thus giving a decisive confirmation about the absence of intrusion variables (Case 1(i) of the rule of thumb). The sequence $B_n$ is declining and thus asymptotically bounded, which strengthens the conclusion even more (see the second part of Case 1(ii) of the rule of thumb).

Figure~\ref{fig-noise}
\begin{figure}[t!]
\centering
\subfigure[$I_n$ when $\alpha=2$, $\beta=3$.]{%
\includegraphics[height=0.3\textwidth]{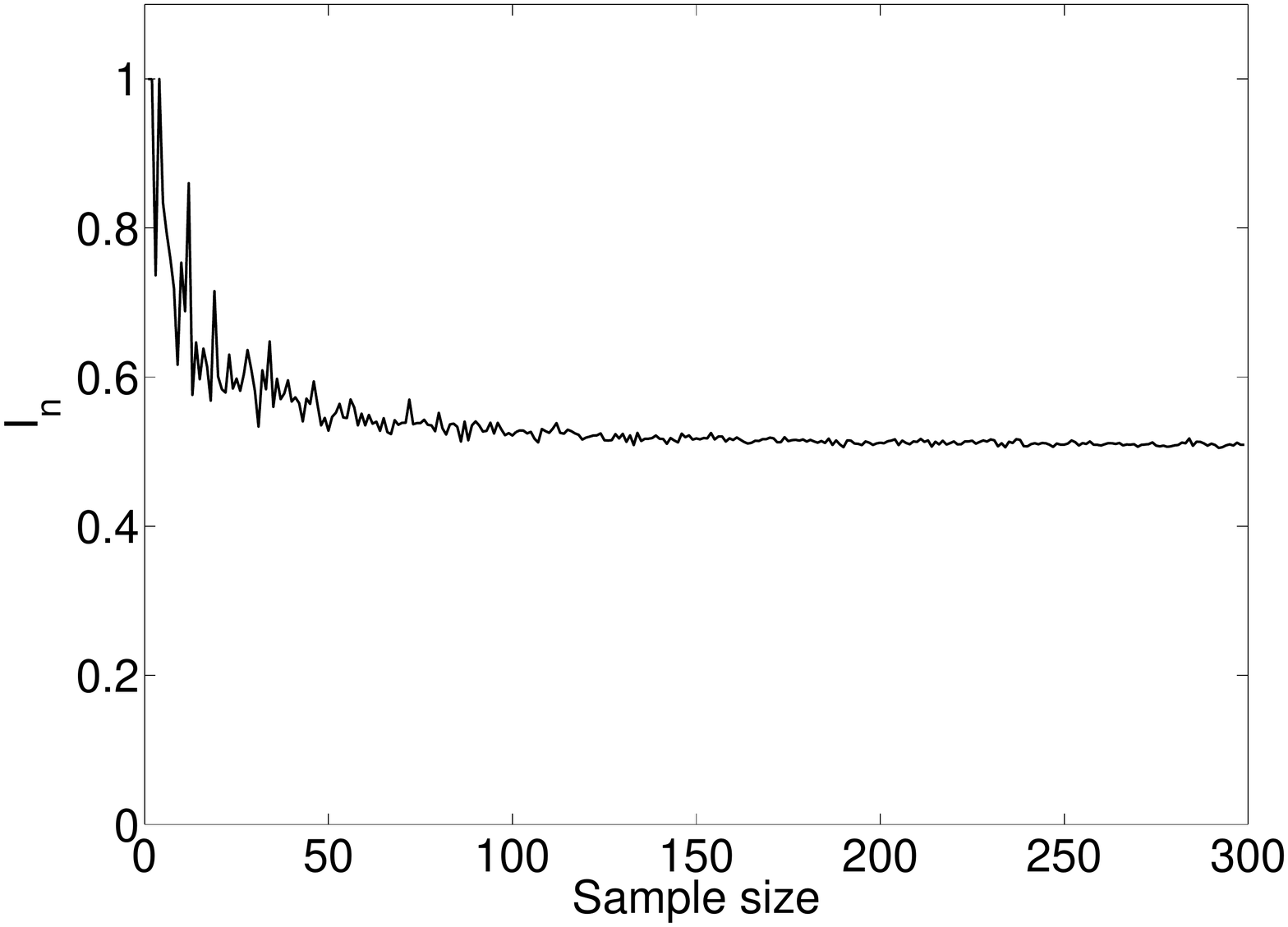}}
\hfill
\subfigure[$B_n$ when $\alpha=2$, $\beta=3$.]{%
\includegraphics[height=0.3\textwidth]{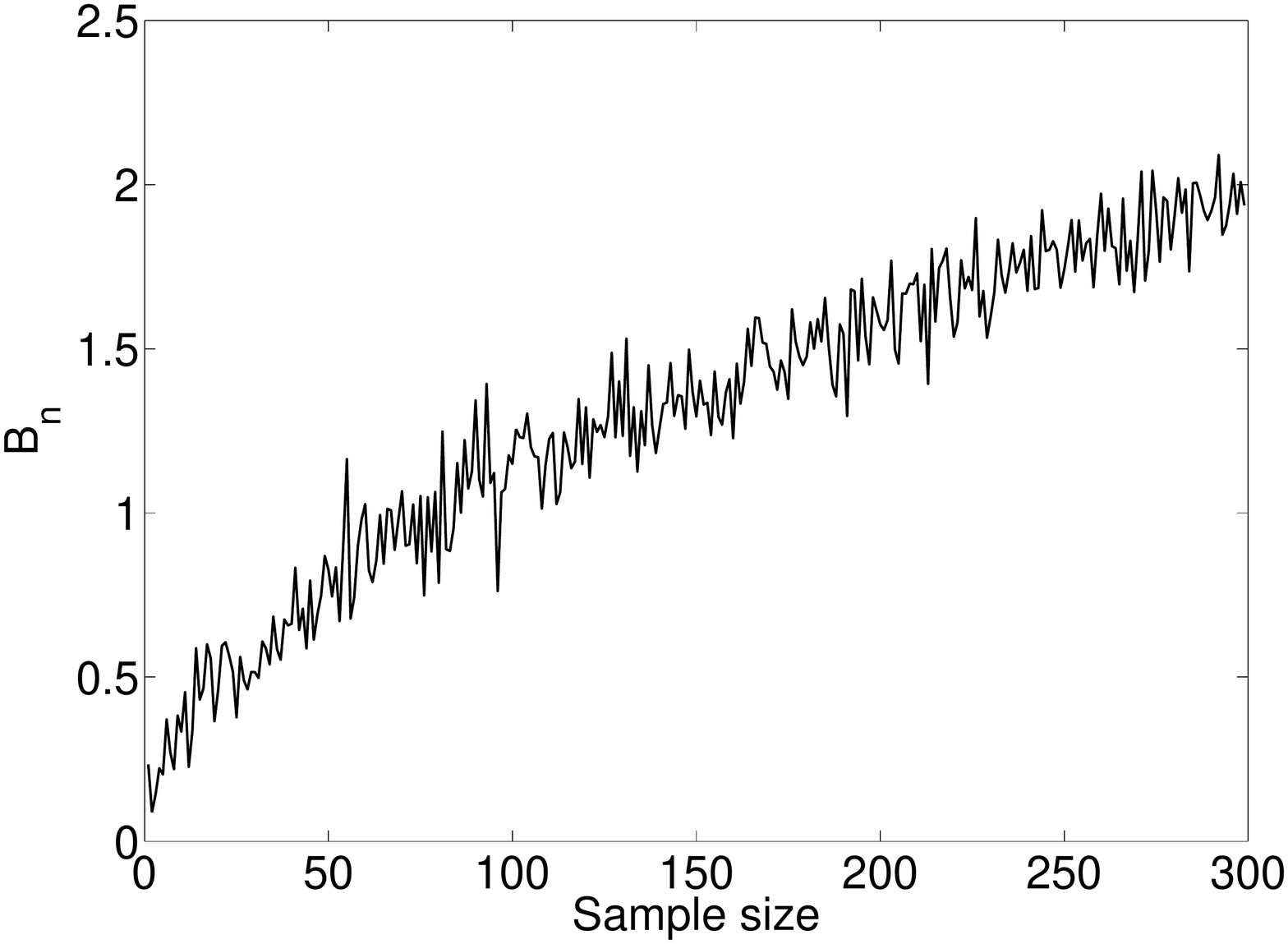}}
\\
\subfigure[$I_n$ when $\alpha=2$, $\beta=2$.]{%
\includegraphics[height=0.3\textwidth]{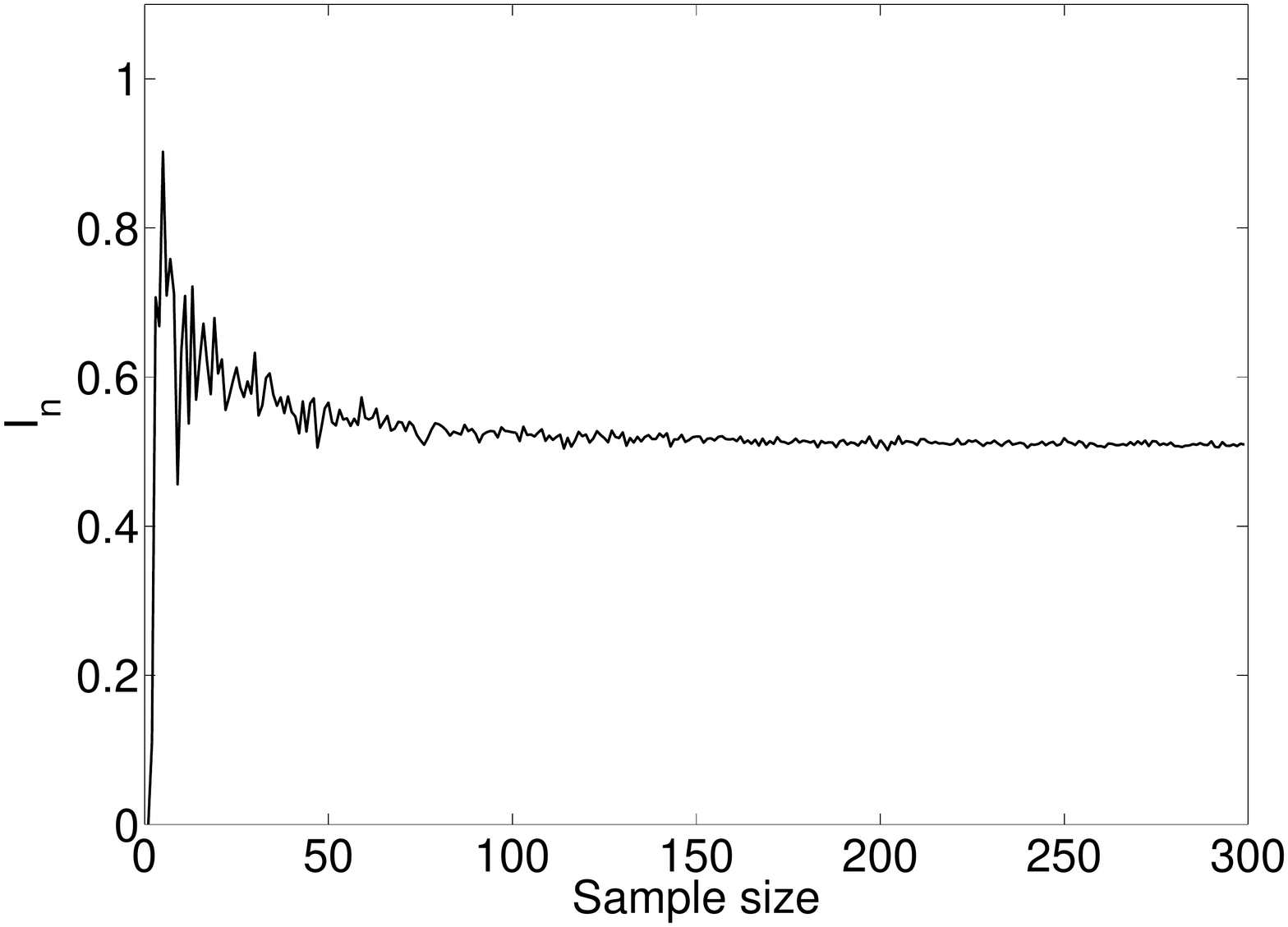}}
\hfill
\subfigure[$B_n$ when $\alpha=2$, $\beta=2$.]{%
\includegraphics[height=0.3\textwidth]{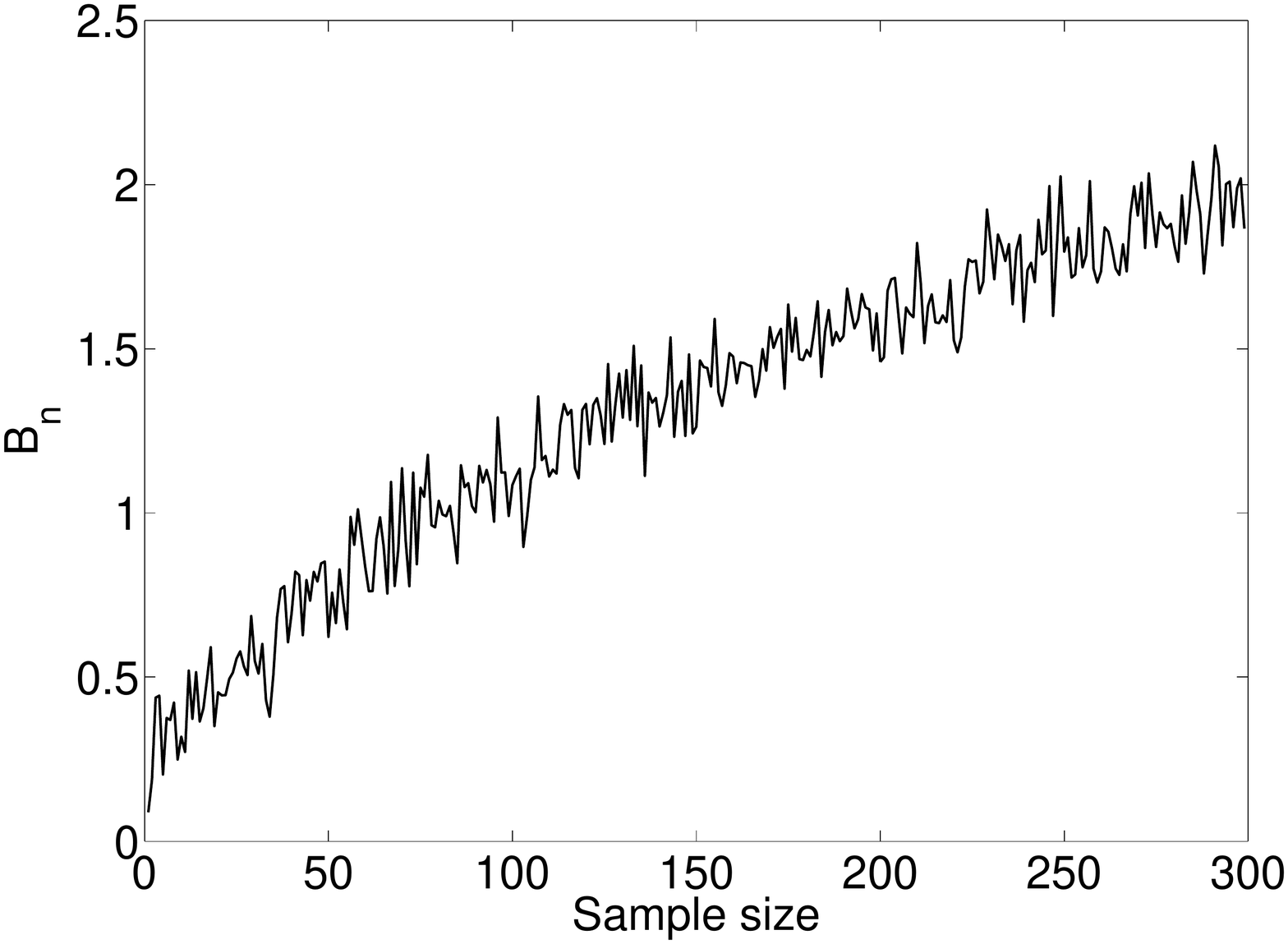}}
\\
\subfigure[$I_n$ when $\alpha=3$, $\beta=2$.]{%
\includegraphics[height=0.3\textwidth]{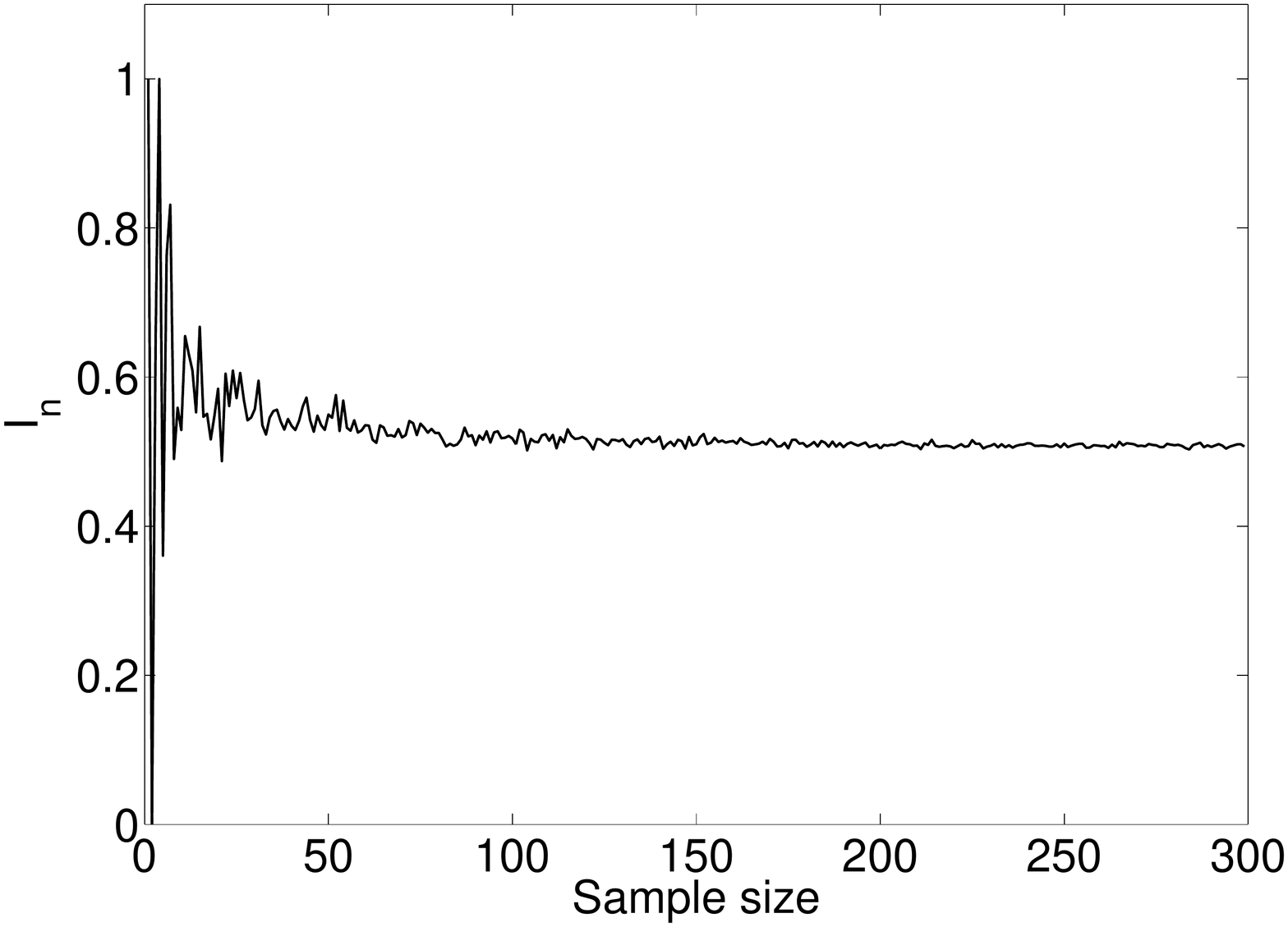}}
\hfill
\subfigure[$B_n$ when $\alpha=3$, $\beta=2$.]{%
\includegraphics[height=0.3\textwidth]{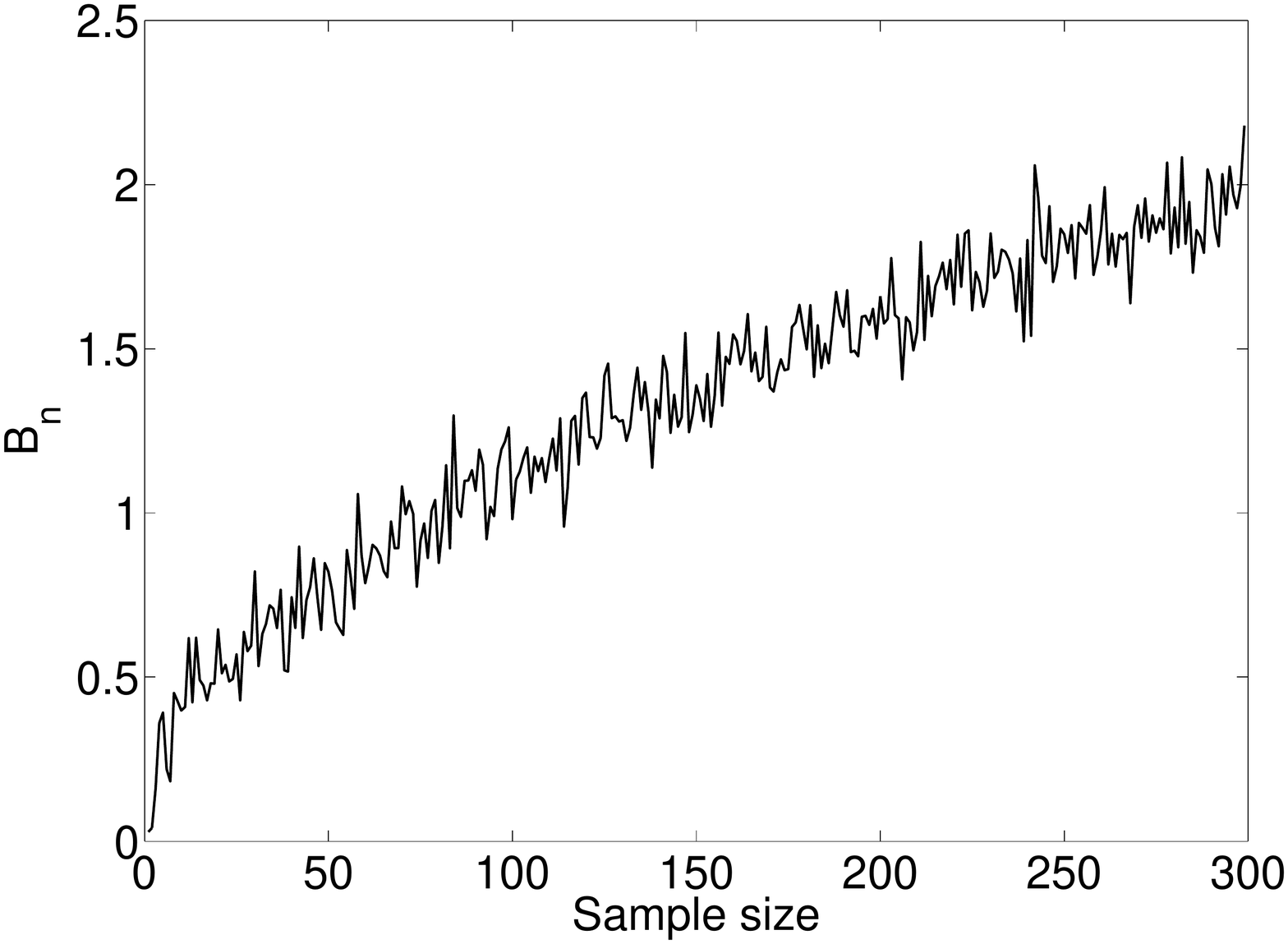}}
\caption{$I_n$ and $B_n$ with respect to the sample size $1\le n \le 300$ when $X_i\sim \mathrm{Beta}(\alpha,\beta)$,  $h(x)=1-(x-4/5)^2$ on $[0,1]$, and compromised outputs with Gaussian intrusions $\varepsilon_i\sim N(0,\sigma^2_{\varepsilon})$ with $\sigma^2_{\varepsilon}=0.01$.}
\label{fig-noise}
\end{figure}
illustrates the rule of thumb when the outputs are compromised. Clearly, the sequence $I_n$ tends to $1/2$, and it does so in a fairly rapid fashion (Case 2 of the rule). Next, we check the behaviour of the sequence $B_n$, which is definitely growing, and we thus conclude (Case 2(i) of the rule) that intrusion variables are present.

Of course, as is the case with all rules of thumb, they generally work, but there are situations when they fail, due to the simple reason that they rely on theoretical results, whose validity is usually based on certain assumptions. Hence, to rigorously sort out when the above rule of thumb works and when it does not, we next describe the results and assumptions that give rise to the rule. We start with the note that by having assumed the continuity of the cdf $F$, we have excluded (almost surely) all the ties among the inputs $X_1,\dots , X_n$, and thus given rise to the order statistics $X_{1:n}< \cdots < X_{n:n}$ that uniquely determine the corresponding concomitants $Y_{1,n}, \dots, Y_{n,n}$ (e.g., David and Nagaraja,~2003).

\begin{theorem}\label{Thm-A}
Let the non-compromised outcomes be in reasonable order, that is, let the transfer function $h(x)$ and the inputs $X_1,\dots , X_n$ satisfy Definition~\ref{order}. If $B_n$ is asymptotically bounded, then the cdf $F_{\varepsilon} $ is degenerate at $0$.
\end{theorem}

\begin{proof}
We first provide additional information about the concomitants $Y_{1,n}, \dots, Y_{n,n}$ under the above specified model: the inputs $X_1,\dots , X_n \sim F$ are iid, the intrusion variables $\varepsilon_1,\dots , \varepsilon_n\sim F_{\varepsilon}$ are also iid, and the two sets of random variables are independent, that is, the inputs and the intrusion variables are mutually independent. Under this model, the concomitants admit the representation
\[
Y_{i,n}=h(X_{i:n})+\varepsilon_{[i]}, \quad i=1,\dots , n,
\]
where $\varepsilon_{[i]}$ denotes the particular intrusion variable associated with $X_{i:n}$. Equipped with these facts, we write the bounds
\begin{align*}
B_n&\ge {1\over \sqrt{n}} \sum_{i=2}^n | \varepsilon_{[i]}-\varepsilon_{[i-1]} |
-{1\over \sqrt{n}} \sum_{i=2}^n \big | h(X_{i:n})-h(X_{i-1:n}) \big |
\\
&\ge  {1\over \sqrt{n}} \sum_{i=2}^n \Big (| \varepsilon_{[i]}-\varepsilon_{[i-1]} |
- \mathbf{E}\big [| \varepsilon_{[i]}-\varepsilon_{[i-1]} |\big ]\Big )
+ {n-1\over \sqrt{n}} \mathbf{E}\big [| \varepsilon_{[i]}-\varepsilon_{[i-1]} |\big ]
+ O_{\mathbf{P}}(1)
\end{align*}
when $n\to \infty $. It is known (e.g., David and Nagaraja,~2003, p.~145) that  $\varepsilon_{[1]},\dots , \varepsilon_{[n]}$ are iid and follow the same  cdf $ F_{\varepsilon}$ as the original intrusion variables $\varepsilon_1,\dots , \varepsilon_n$.
Since we have assumed that the latter ones have finite second moments, we therefore conclude that
\begin{align*}
B_n
&\ge \sqrt{n}~ \mathbf{E}\big [|\varepsilon_{[2]}-\varepsilon_{[1]}|\big ]+ O_{\mathbf{P}}(1)
\\
&= \sqrt{n}~ \mathbf{E}\big [|\varepsilon_{2}-\varepsilon_{1}|\big ]+ O_{\mathbf{P}}(1)
\end{align*}
when $n\to \infty $. Since we have assumed that the sequence $B_n$ is asymptotically bounded when $n\to \infty $, the expectation $\mathbf{E} [| \varepsilon_{2}-\varepsilon_{1} |  ]$ must be equal to $0$. The latter can be true only if the cdf $F_{\varepsilon}$ is degenerate at some point, which of course must be equal to $0$ due to the assumed zero means of the intrusion variables. This finishes the proof of Theorem~\ref{Thm-A}.
\end{proof}

Hence, Theorem~\ref{Thm-A} says that if the sequence $B_n$ is asymptotically bounded  when $n\to \infty $, then there are no intrusion variables, as illustrated in  Figure~\ref{fig-pure}. The following theorem explains what happens when $B_n$ is not asymptotically bounded, which is the case explored in Figure~\ref{fig-noise}.

\begin{theorem}[Gribkova and Zitikis, 2018]
\label{Thm-B}
Let the non-compromised outputs $h(X_i)$ have finite second moments, and let the intrusion variables $\varepsilon_{1},\dots , \varepsilon_{n}$, irrespective of whether they are degenerate or not, also have finite second moments. If $B_n$ grows to infinity when $n\to \infty $, then $I_n$ converges to $ 1/2$.
\end{theorem}

Hence, Theorem~\ref{Thm-B} says that if $B_n$ grows to infinity, then the control system is compromised, provided that the limit $1/2$ indicates the presence of intrusion variables, which indeed happens most of the time, but not always, and this is the very point where the condition $h(a)\ne h(b)$ arises. The following theorem and discussion clarify the matter.

\begin{theorem}[Gribkova and Zitikis, 2018]
\label{Thm-C}
Let the derivative $(\mathrm{d}/\mathrm{d}u)h\circ F^{-1}(u)$ exist, be continuous, and not identically equal to $0$ on the interval $[0,1]$. If the outputs are not compromised, that is, if the cdf $F_{\varepsilon} $ is degenerate (at the point $0$), then $I_n$ converges to
\[
I(h):={\int_{a}^{b} (h'(u))_{+}\mathrm{d}u \over \int_{a}^{b}|h(u)|\mathrm{d}u}.
\]
\end{theorem}

Hence, when there are no intrusion variables, $I_n$ converges to $I(h)$, which can be equal to $1/2$. This can happen if and only if $\int_{a}^{b} (h'(u))_{+}\mathrm{d}u=\int_{a}^{b} (h'(u))_{-}\mathrm{d}u$,
where we have used the notation $x_{-}=\max\{-x,0\}$. Since $x=x_{+}-x_{-}$, the above equation is equivalent to $\int_{a}^{b} h'(u)\mathrm{d}u=0$, which means 
\begin{equation}\label{equality}
h(a)=h(b). 
\end{equation}
Consequently, by eliminating the possibility of having equation (\ref{equality}), we establish an one-to-one relationship between the convergence of $I_n$ to $1/2$ and the presence of intrusion variables. We note in this regard that, as far as we are aware of, practically relevant transfer functions are outside the class of those satisfying the relationship $h(a)=h(b)$, unless the control system is down and the transfer function $h(x)$ takes the same value irrespective of $x\in [a,b]$.

To illustrate the case $h(a)=h(b)$ visually, in Figure~\ref{fig-contam-confuse}
\begin{figure}[t!]
\centering
\subfigure[$I_n$ when $\alpha=2$, $\beta=3$.]{%
\includegraphics[height=0.3\textwidth]{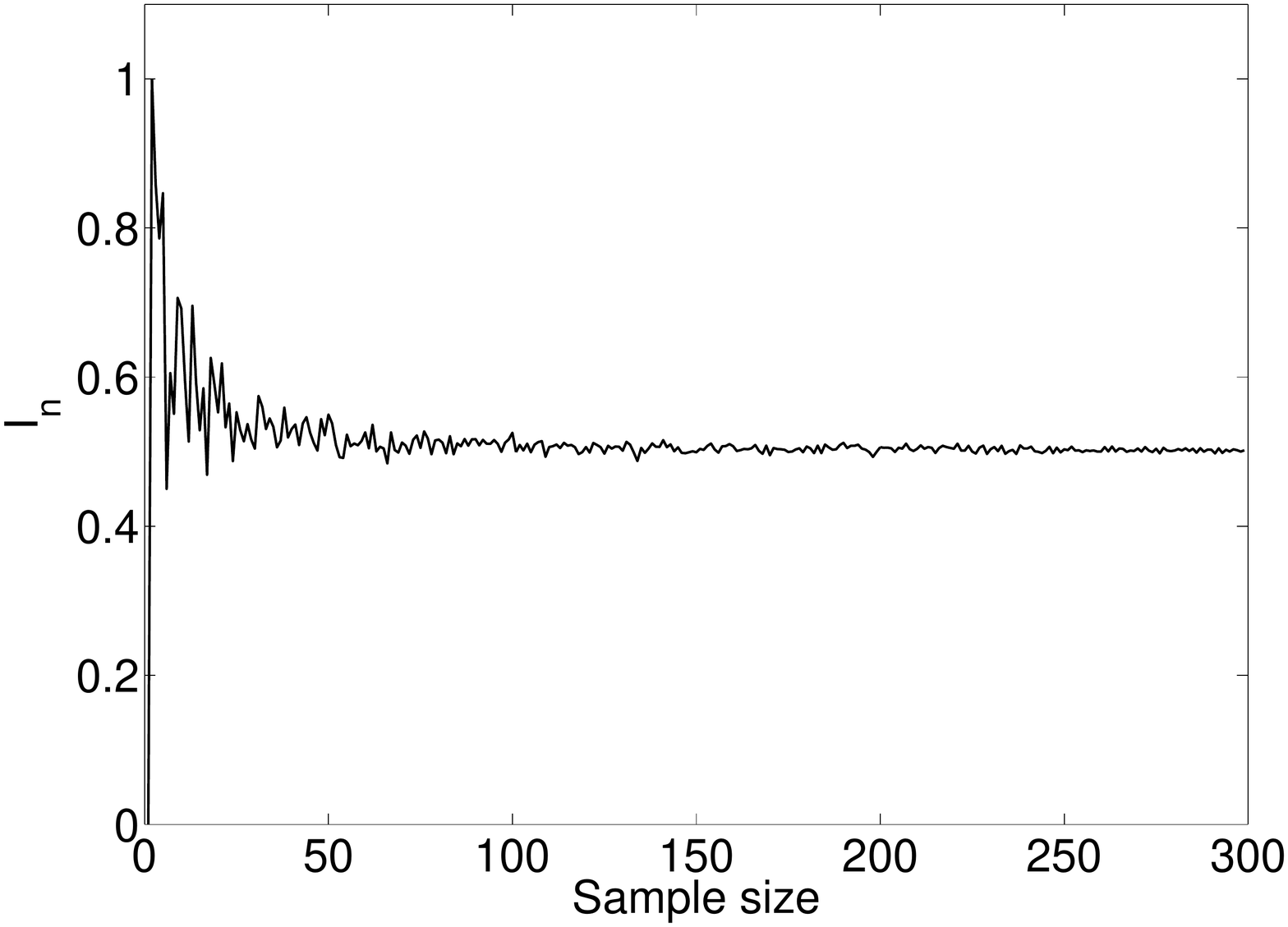}}
\hfill
\subfigure[$B_n$ when $\alpha=2$, $\beta=3$.]{%
\includegraphics[height=0.3\textwidth]{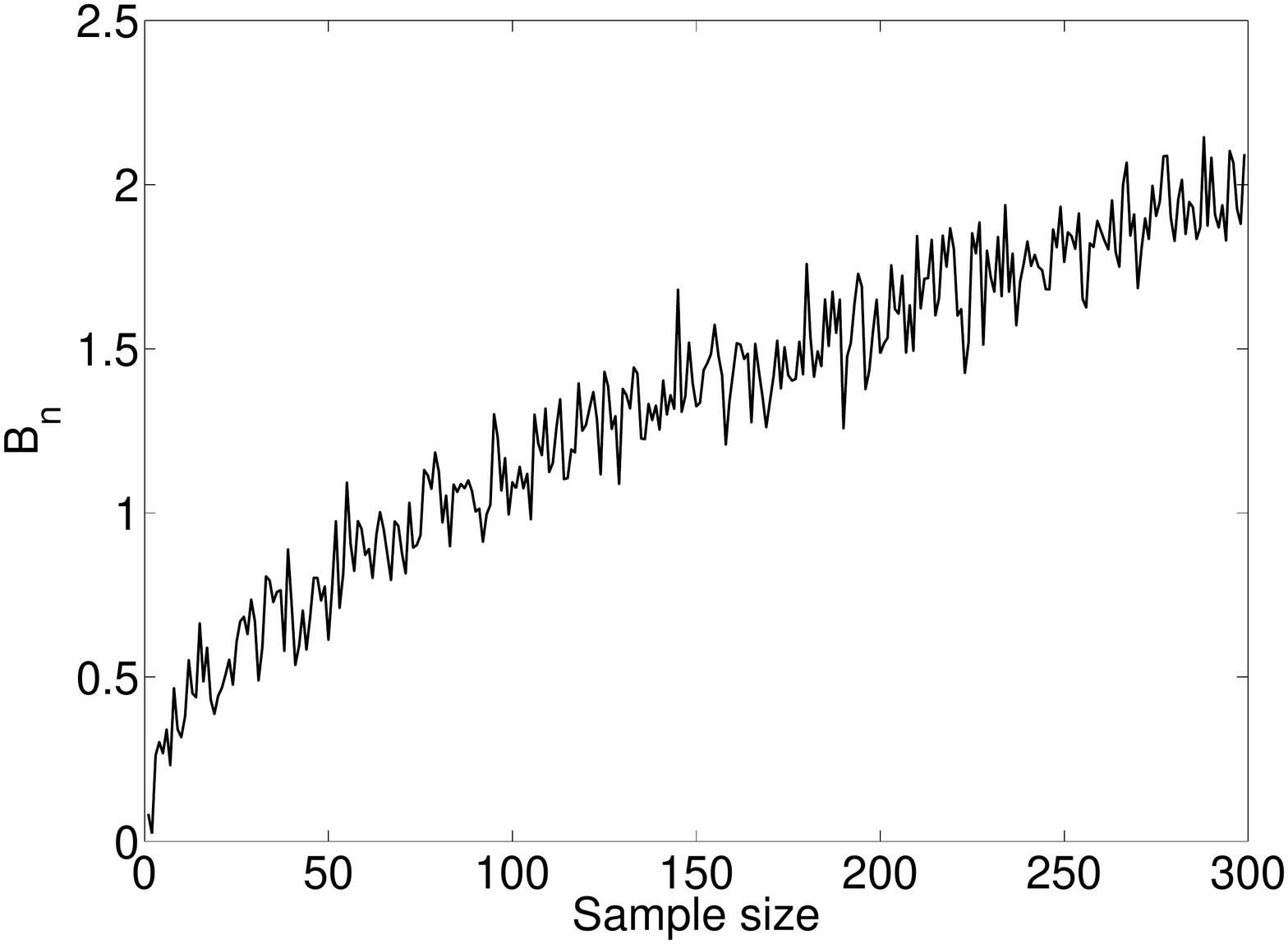}}
\\
\subfigure[$I_n$ when $\alpha=2$, $\beta=2$.]{%
\includegraphics[height=0.3\textwidth]{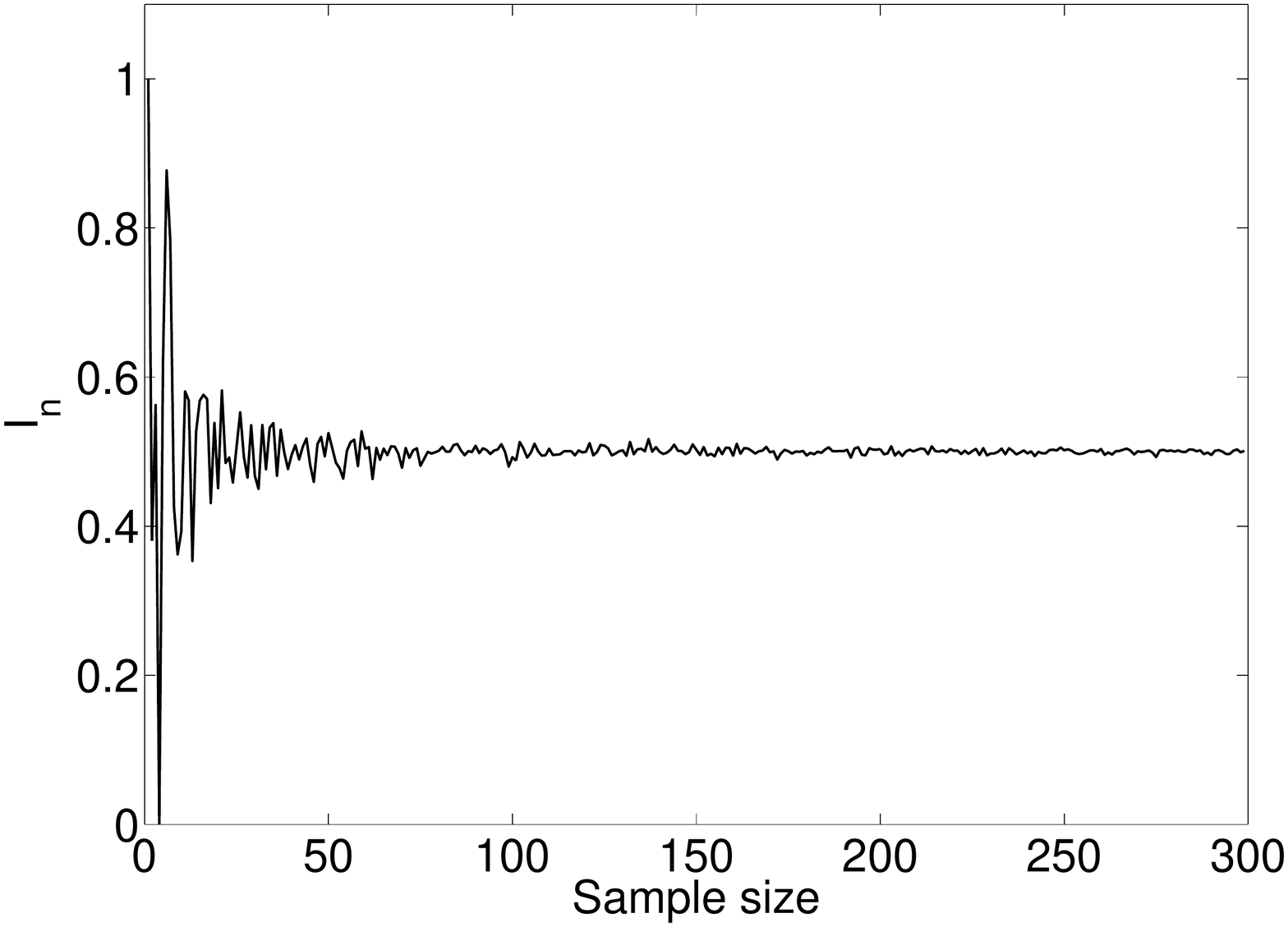}}
\hfill
\subfigure[$B_n$ when $\alpha=2$, $\beta=2$.]{%
\includegraphics[height=0.3\textwidth]{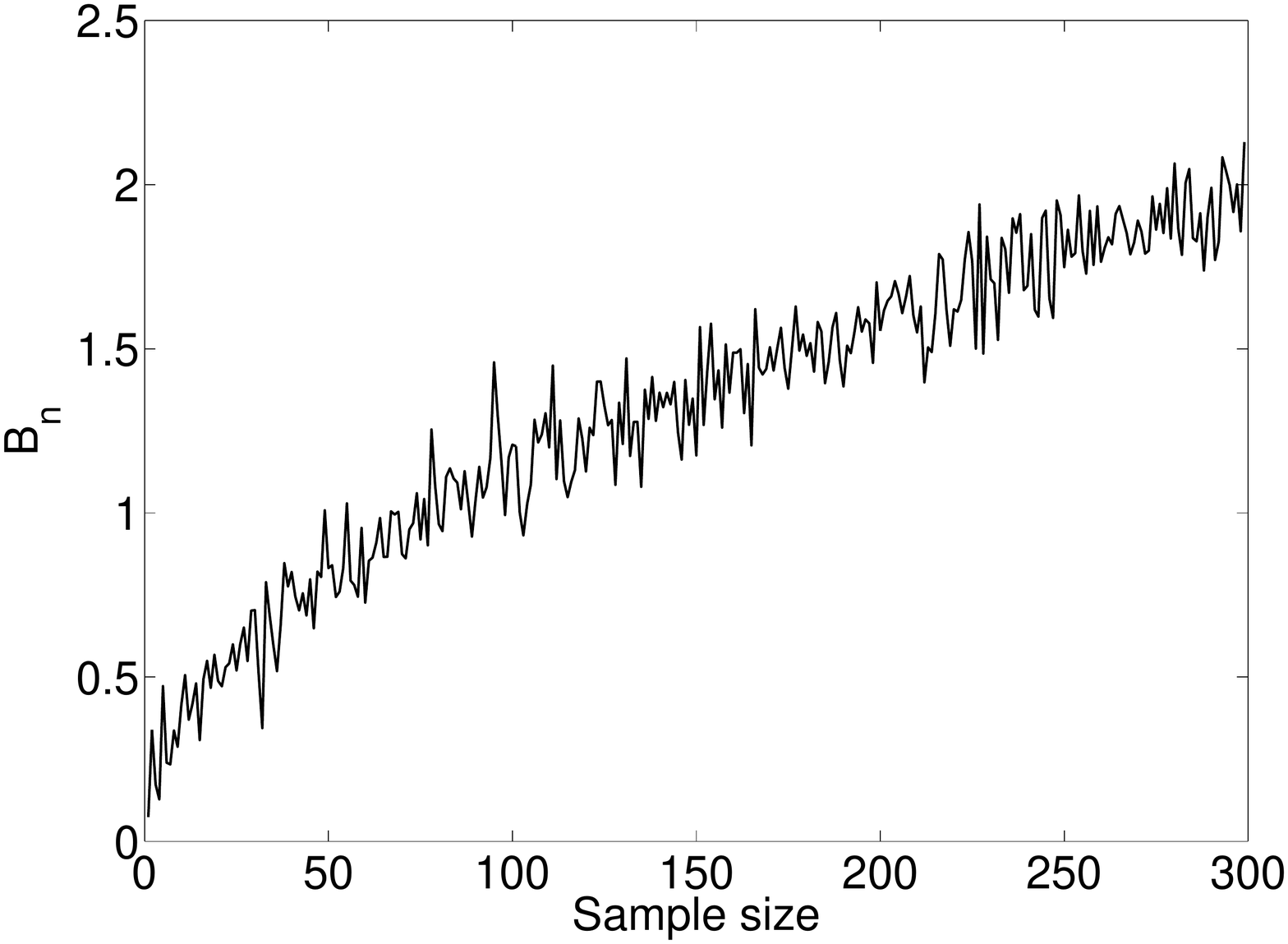}}
\\
\subfigure[$I_n$ when $\alpha=3$, $\beta=2$.]{%
\includegraphics[height=0.3\textwidth]{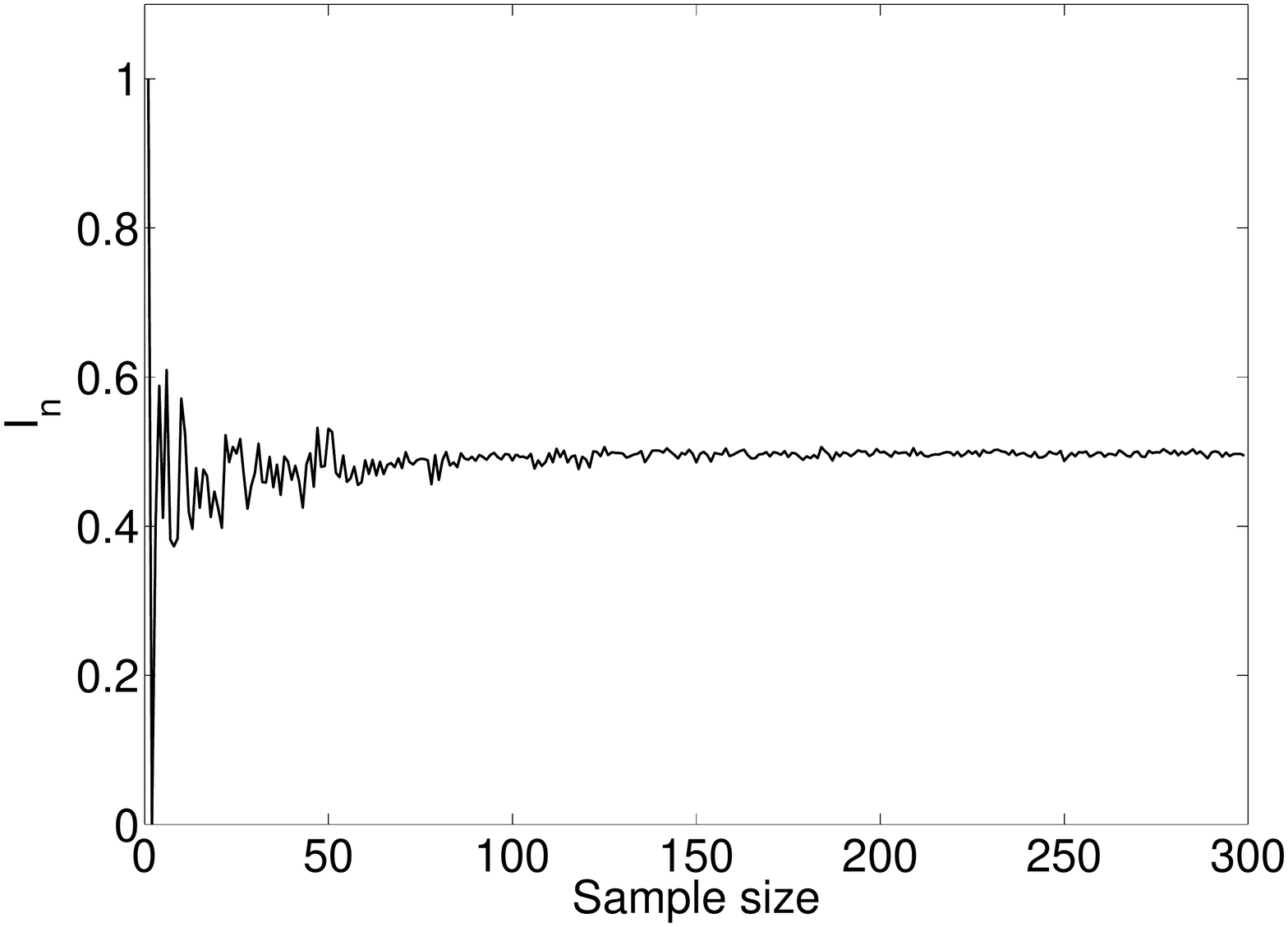}}
\hfill
\subfigure[$B_n$ when $\alpha=3$, $\beta=2$.]{%
\includegraphics[height=0.3\textwidth]{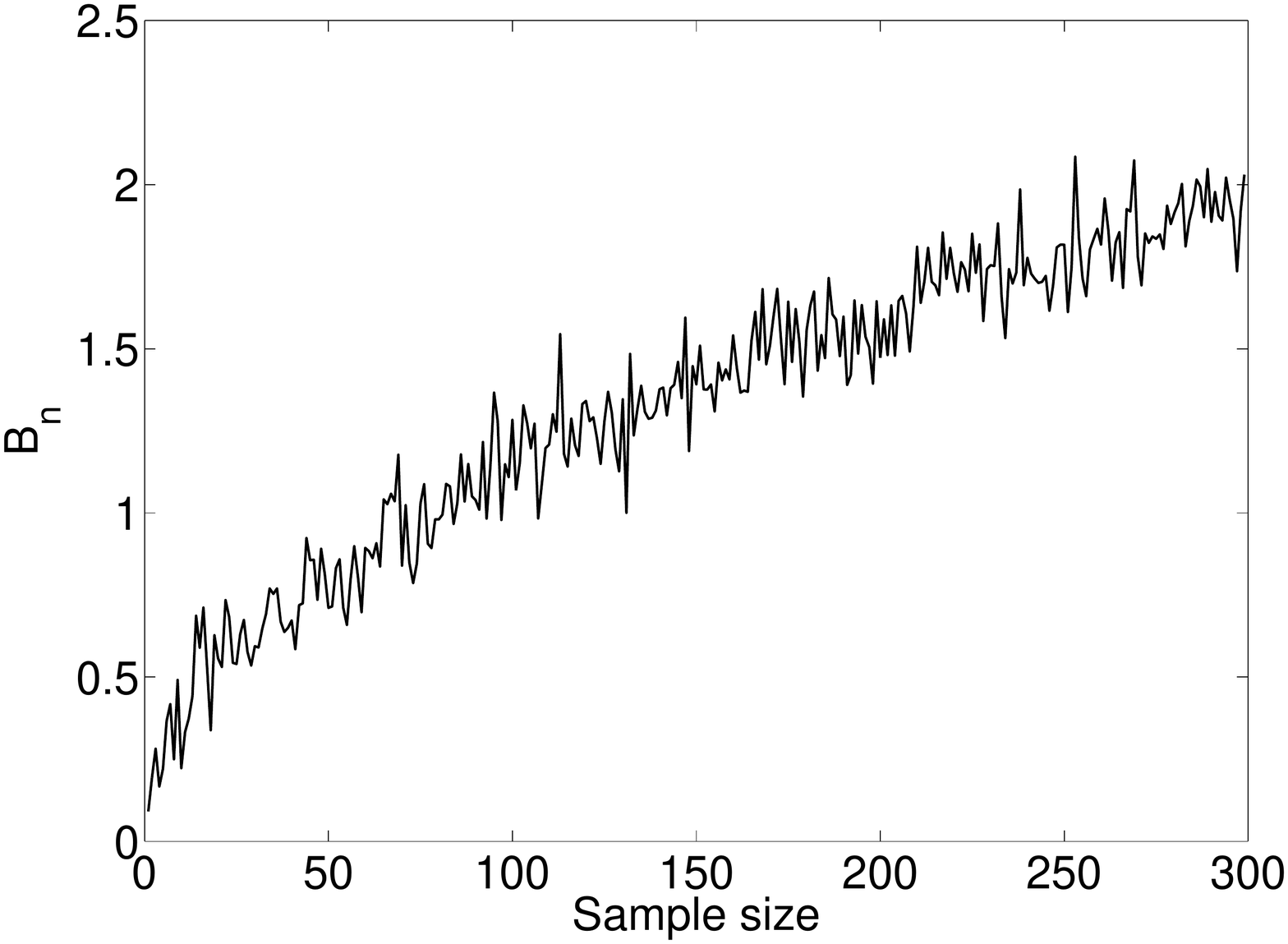}}
\caption{$I_n$ and $B_n$ with respect to the sample size $1\le n \le 300$ when $X_i\sim (8/5)\mathrm{Beta}(\alpha,\beta)$,  $h(x)=1-(x-4/5)^2$ on $[0,8/5]$, and compromised outputs with Gaussian intrusions $\varepsilon_i\sim N(0,\sigma^2_{\varepsilon})$ with $\sigma^2_{\varepsilon}=0.01$.}
\label{fig-contam-confuse}
\end{figure}
we have depicted the asymptotic behaviour of $I_n$ and $B_n$ with added intrusion variables $\varepsilon_1,\dots , \varepsilon_n$. Clearly, $I_n$ tends to $1/2$, and $B_n$ grows. Based on the rule of thumb (see Case 2(i)), we conclude that the systems is compromised by intrusion variables.

For comparison, in Figure~\ref{fig-pure-confuse}
\begin{figure}[t!]
\centering
\subfigure[$I_n$ when $\alpha=2$, $\beta=3$.]{%
\includegraphics[height=0.3\textwidth]{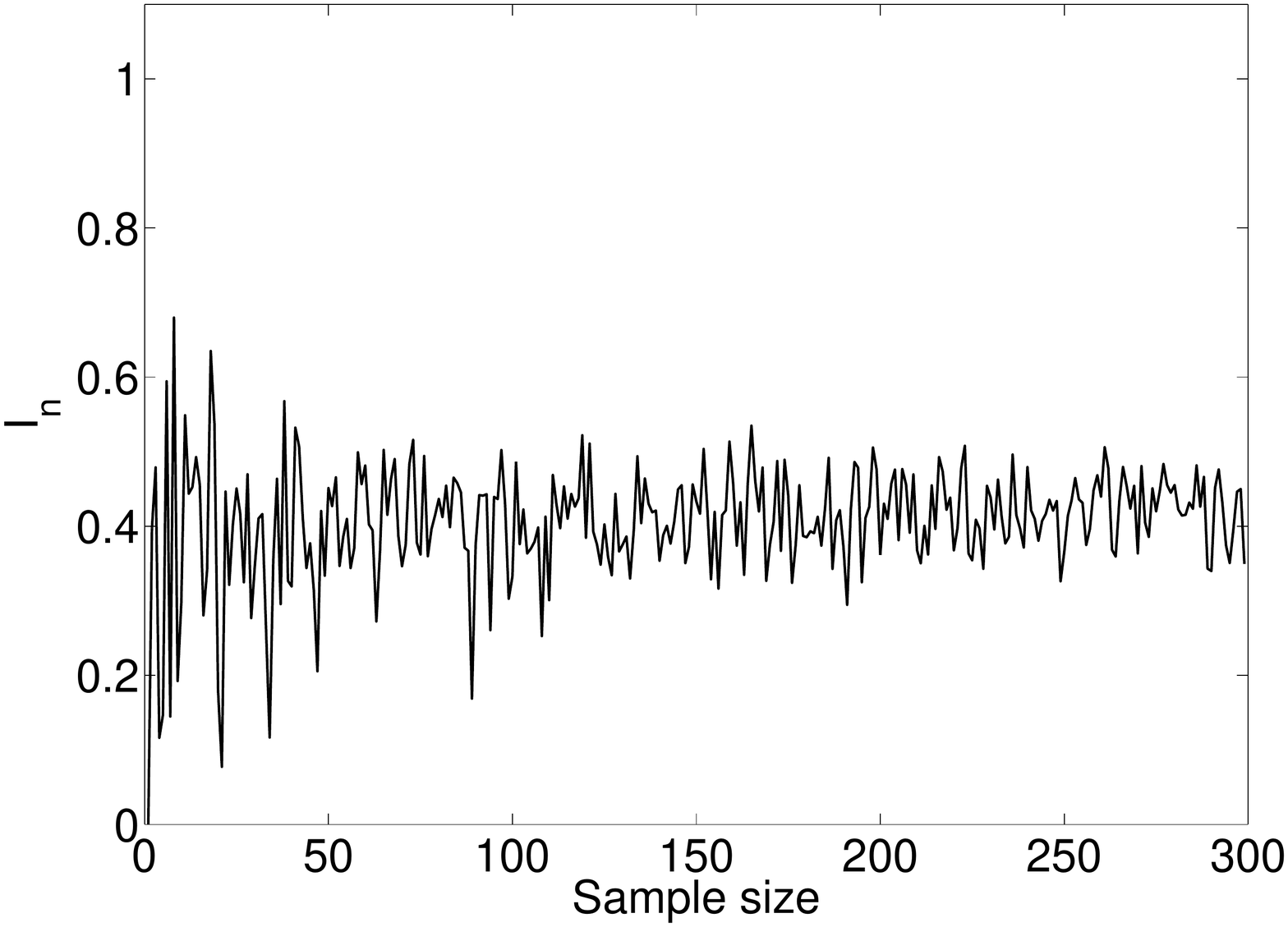}}
\hfill
\subfigure[$B_n$ when $\alpha=2$, $\beta=3$.]{%
\includegraphics[height=0.3\textwidth]{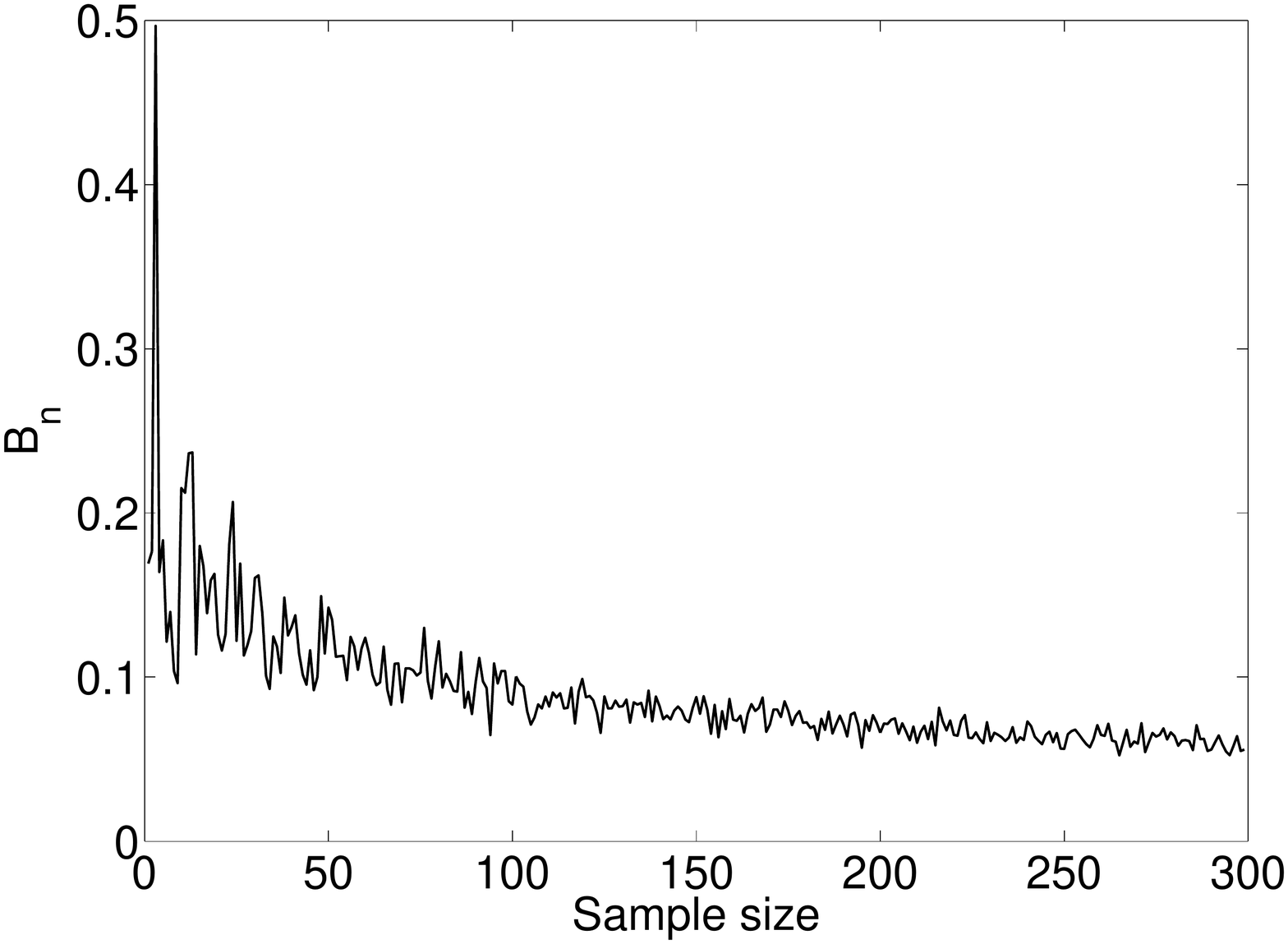}}
\\
\subfigure[$I_n$ when $\alpha=2$, $\beta=2$.]{%
\includegraphics[height=0.3\textwidth]{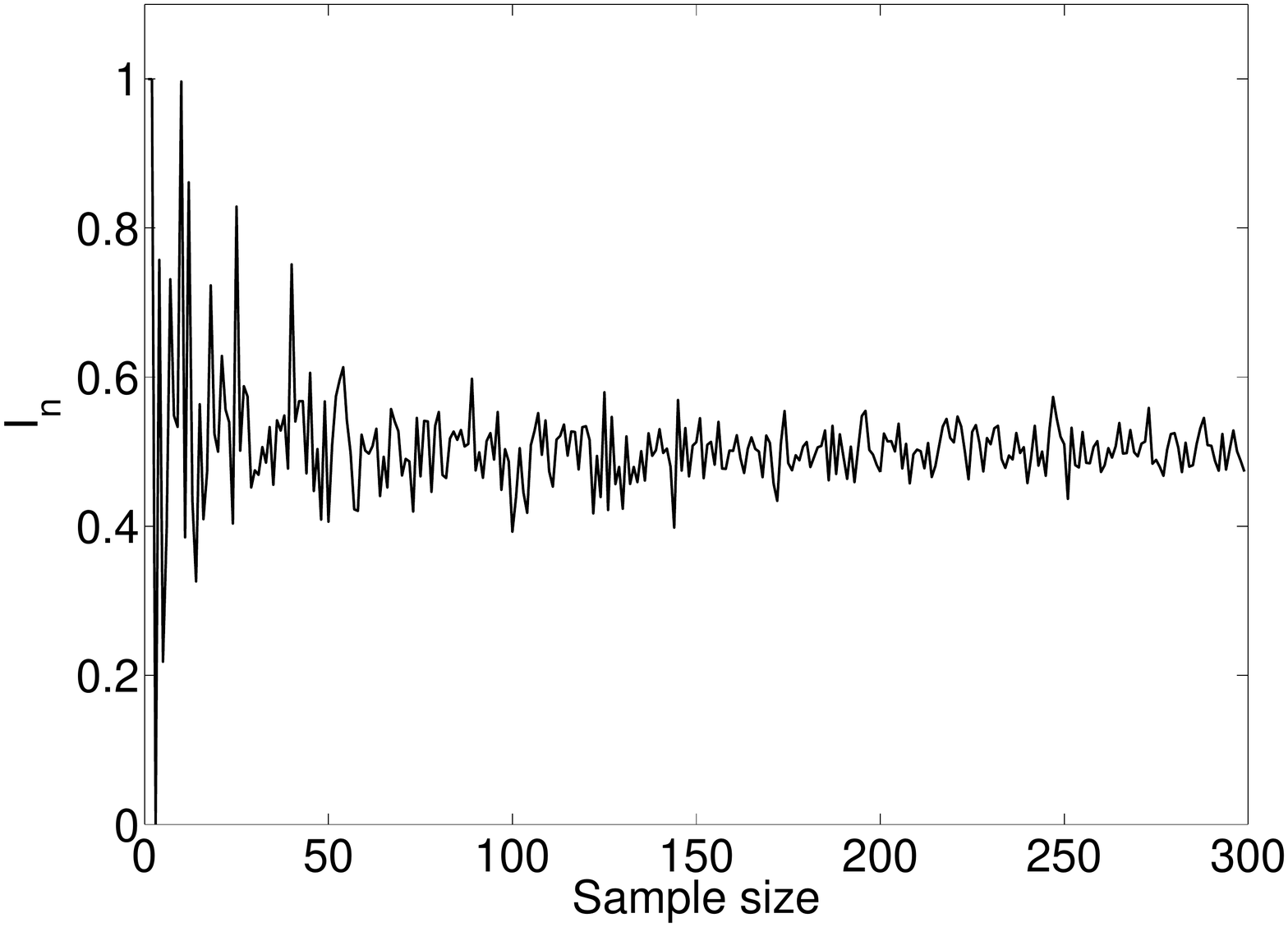}}
\hfill
\subfigure[$B_n$ when $\alpha=2$, $\beta=2$.]{%
\includegraphics[height=0.3\textwidth]{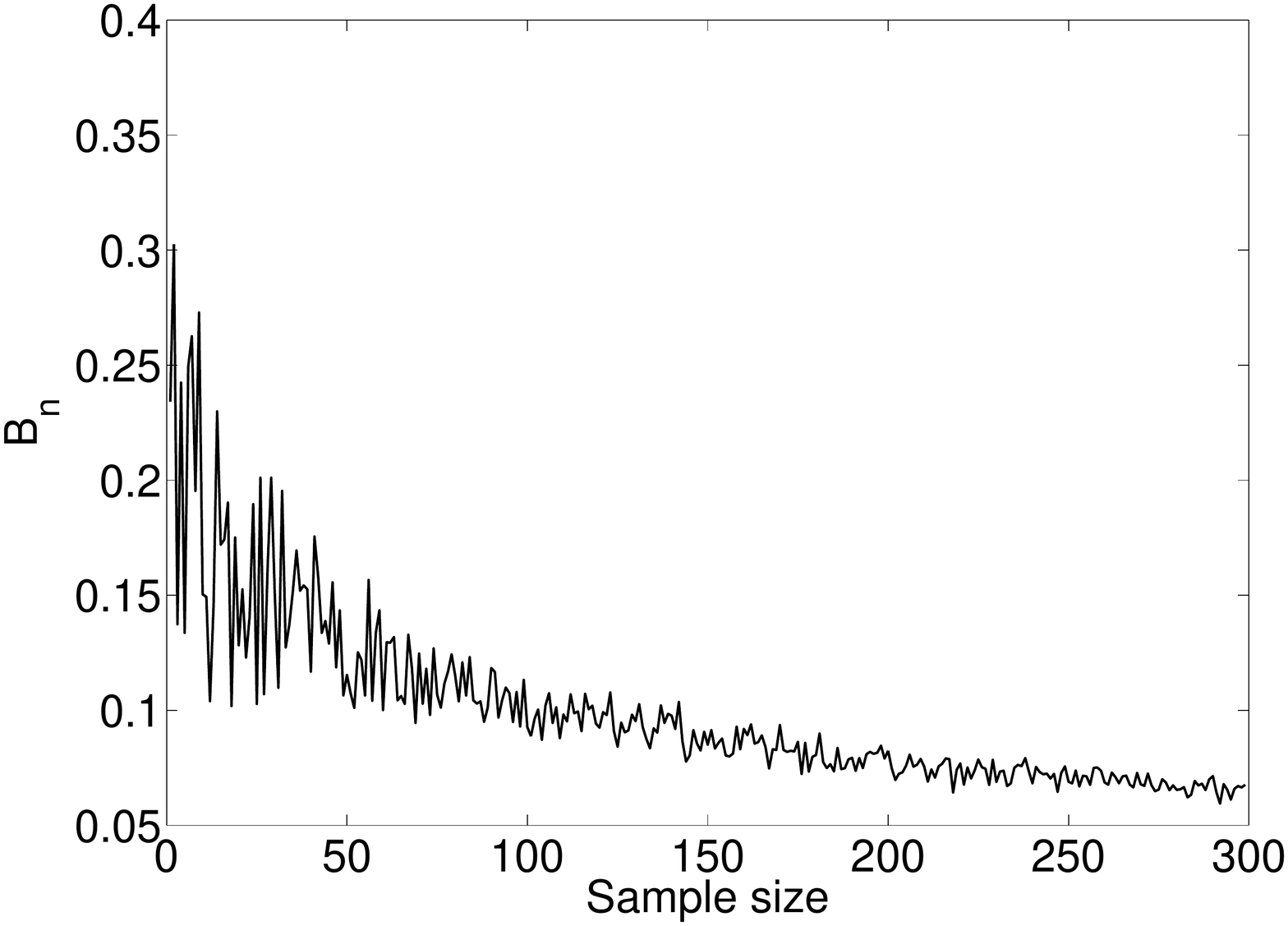}}
\\
\subfigure[$I_n$ when $\alpha=3$, $\beta=2$.]{%
\includegraphics[height=0.3\textwidth]{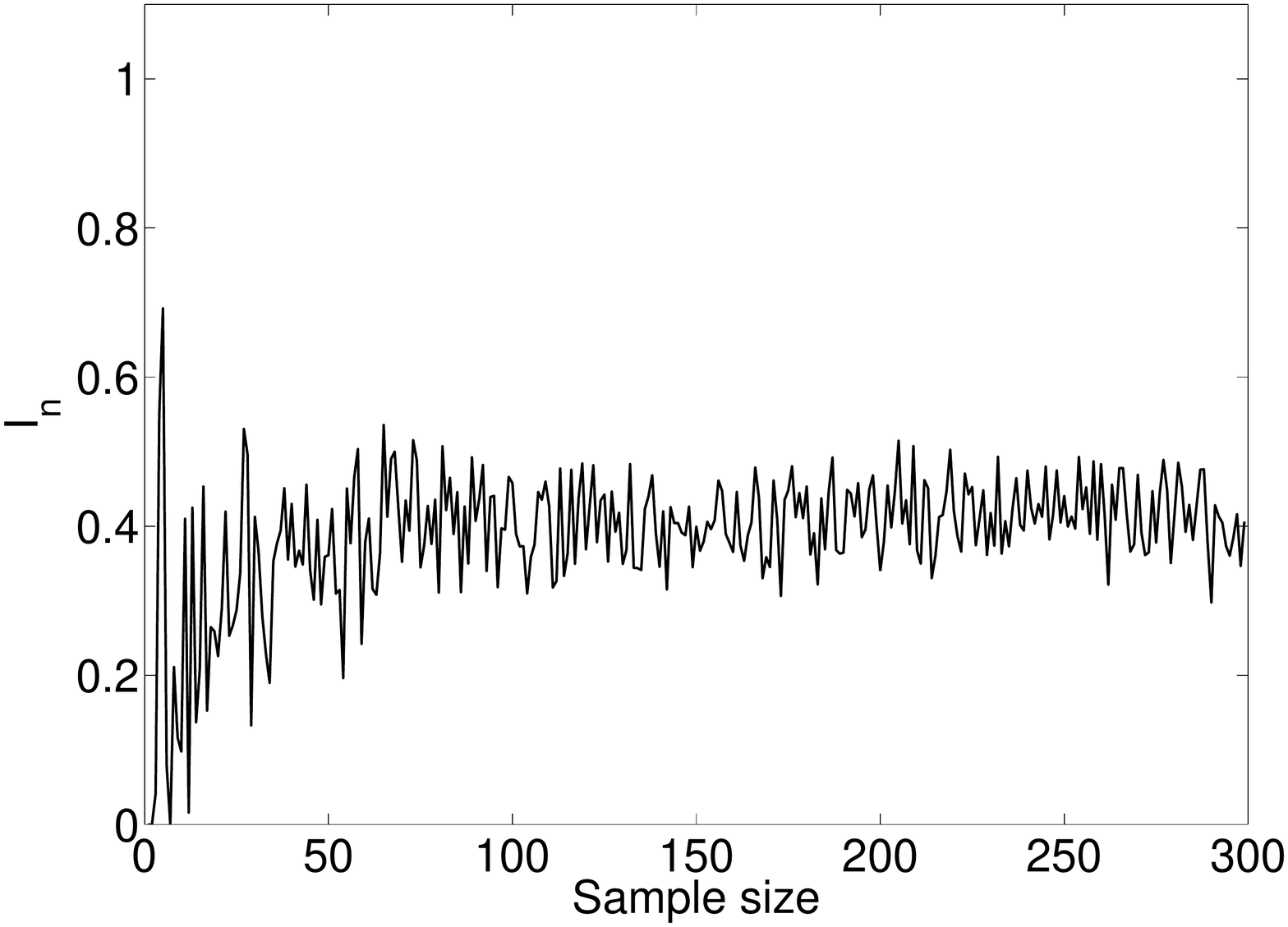}}
\hfill
\subfigure[$B_n$ when $\alpha=3$, $\beta=2$.]{%
\includegraphics[height=0.3\textwidth]{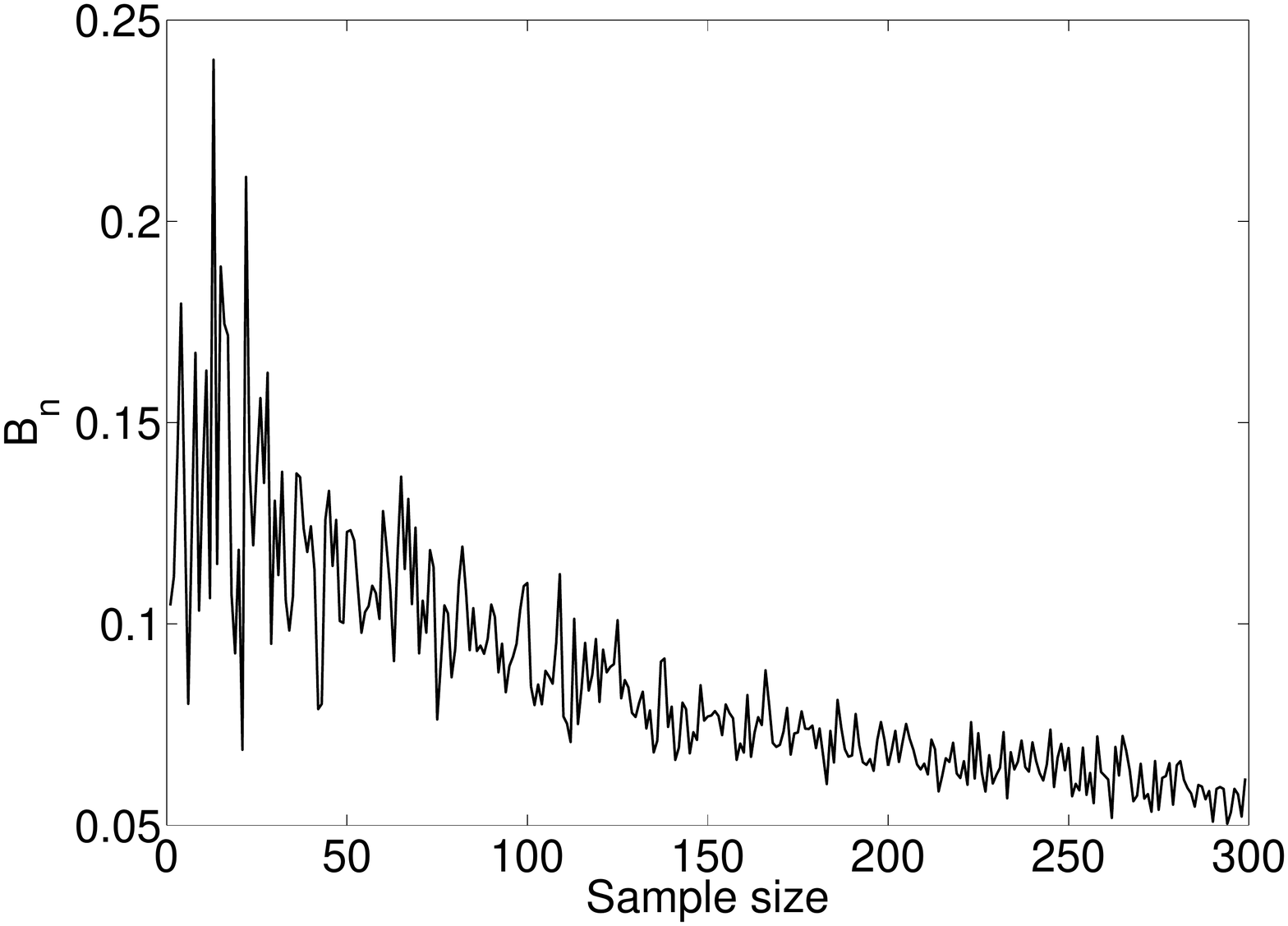}}
\caption{$I_n$ and $B_n$ with respect to the sample size $1\le n \le 300$ when $X_i\sim (8/5)\mathrm{Beta}(\alpha,\beta)$,  $h(x)=1-(x-4/5)^2$ on $[0,8/5]$, and non-compromised outputs.}
\label{fig-pure-confuse}
\end{figure}
we depict asymptotic behaviour of $I_n$ and $B_n$ when the outputs are not compromised. The asymptotic behaviour of $I_n$ is erratic, though perhaps we could still argue that there is some tendency to get closer to $1/2$. Nevertheless, the sequence $B_n$ is clearly declining  and thus asymptotically bounded, and we therefore decisively conclude (Case 2(ii) of the rule of thumb) that the system is not compromised.

We finish the present section by recalling the earlier note (see below (\ref{condition-0b})) that the boundedness of the derivative $(\mathrm{d}/\mathrm{d}u)h\circ F^{-1}(u)$ can be relaxed. This is indeed possible due to another theorem of Gribkova and Zitikis (2018), which is more complex than Theorem~\ref{Thm-C} and allows the derivative $(\mathrm{d}/\mathrm{d}u)h\circ F^{-1}(u)$ to grow at the two endpoints of its domain of definition $[0,1]$. Nevertheless, the simpler case covered by  Theorem~\ref{Thm-C} is sufficiently encompassing and quite attractive from the practical point of view.

\section{Deterministic inputs and vulnerability testing}
\label{determ}

Whether the rule of thumb detects intrusions or not, we can still wish to double-check the finding, as we did near the end of the previous section. Furthermore, there can even be a necessity to test the control system's vulnerability (e.g., Hug and Giampapa, 2012; and references therein). In such cases, instead of the pre-whitened random inputs $X_1,\dots , X_n$, it is natural to feed into the system deterministic inputs, such as
\begin{equation}\label{eq-1c}
x_{i,n}=a+(b-a){i-1 \over n-1}, \quad i=1,\dots , n,
\end{equation}
and then apply the rule of thumb.

As an illustration, we go back to Figures~\ref{fig-contam-confuse} and~\ref{fig-pure-confuse}, which concern the case $h(a)=h(b)$, and reassess our findings using  deterministic inputs (\ref{eq-1c}) with $a=0$ and $b=8/5$. In Figure~\ref{fig-deter},
\begin{figure}[t!]
\centering
\subfigure[$I_n$ when $\sigma^2_{\varepsilon}=0.01$.]{%
\includegraphics[height=0.3\textwidth]{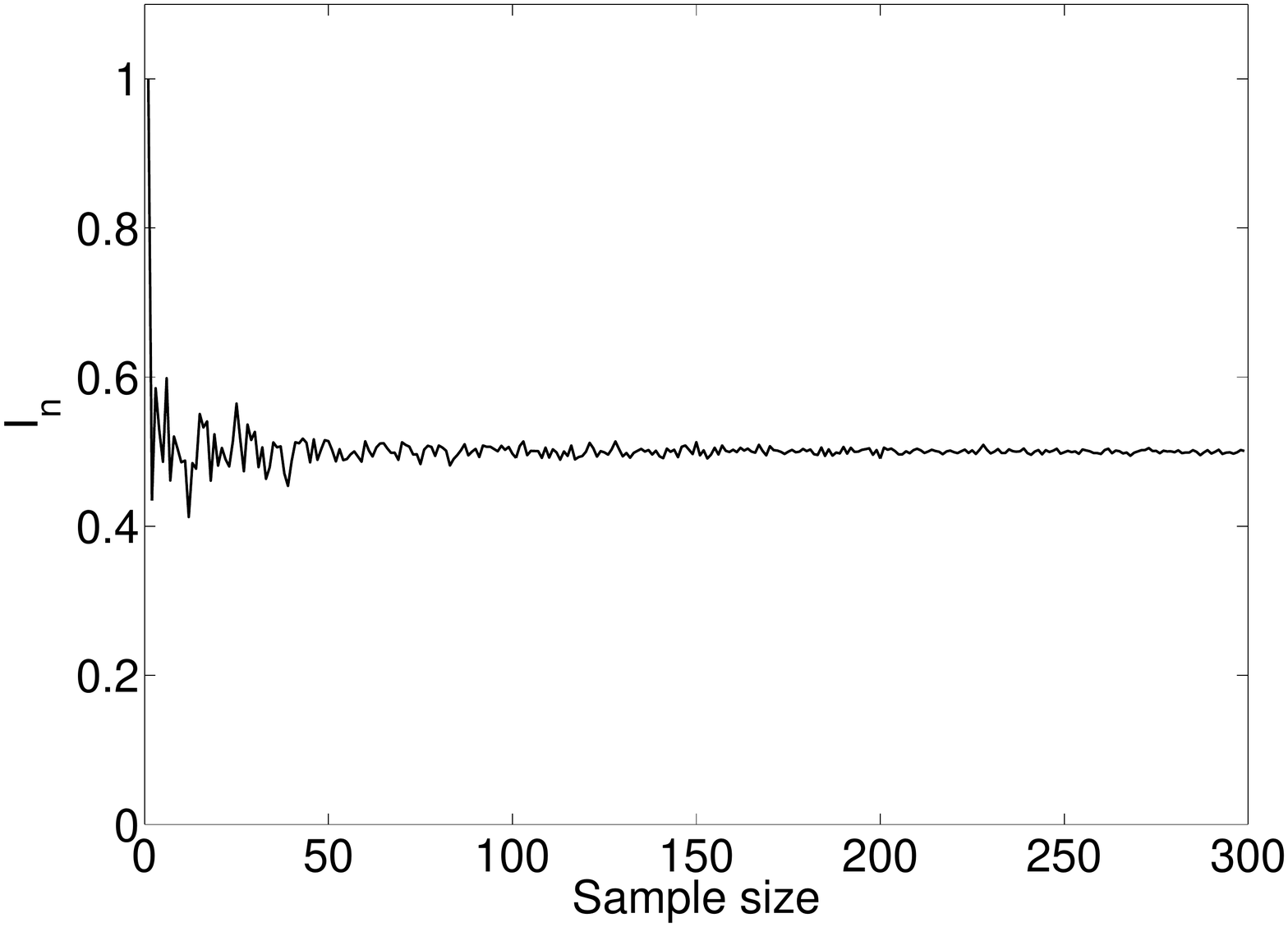}}
\hfill
\subfigure[$B_n$ when $\sigma^2_{\varepsilon}=0.01$.]{%
\includegraphics[height=0.3\textwidth]{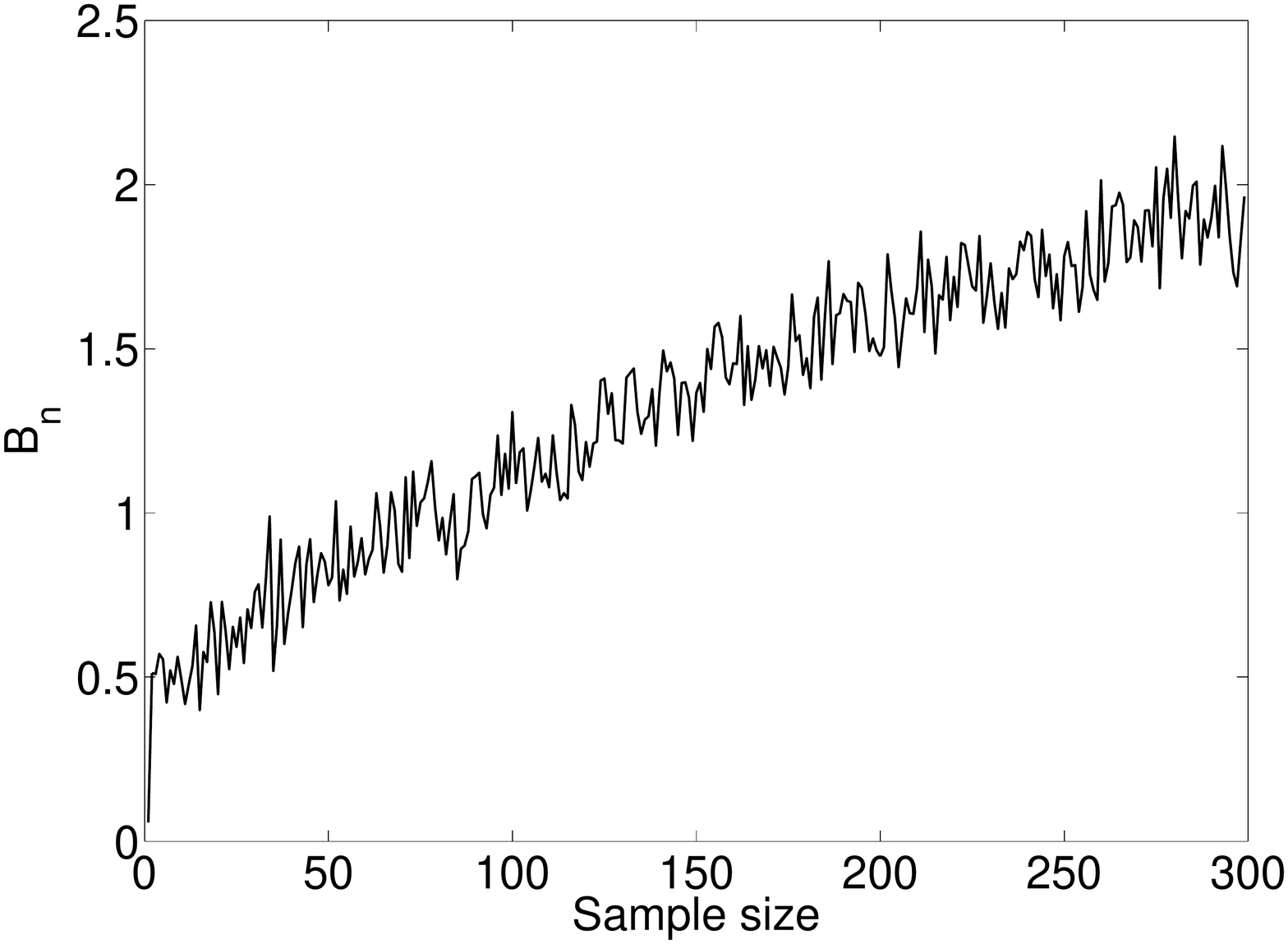}}
\\
\subfigure[$I_n$ when $\sigma^2_{\varepsilon}=0$.]{%
\includegraphics[height=0.3\textwidth]{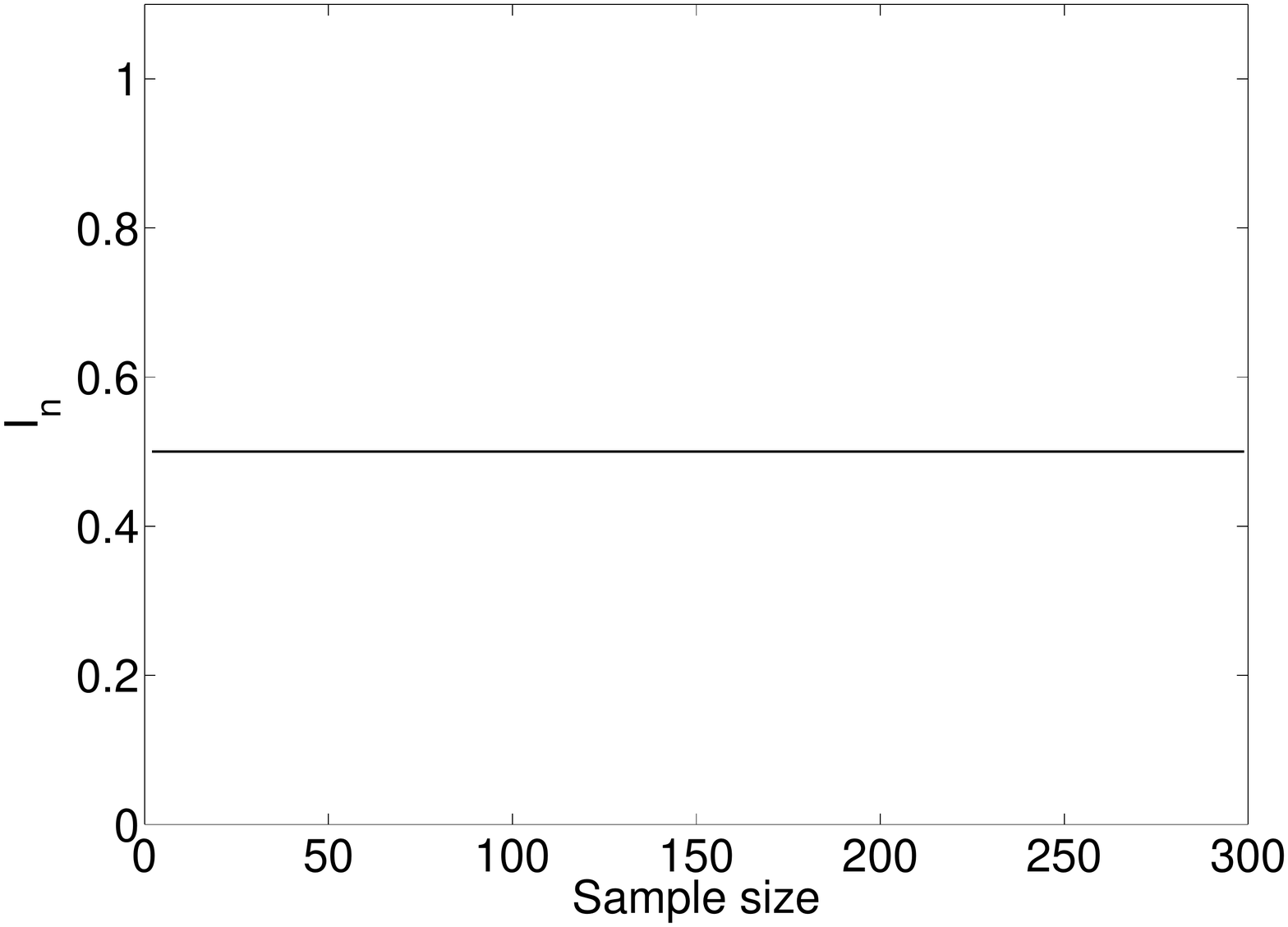}}
\hfill
\subfigure[$B_n$ when $\sigma^2_{\varepsilon}=0$.]{%
\includegraphics[height=0.3\textwidth]{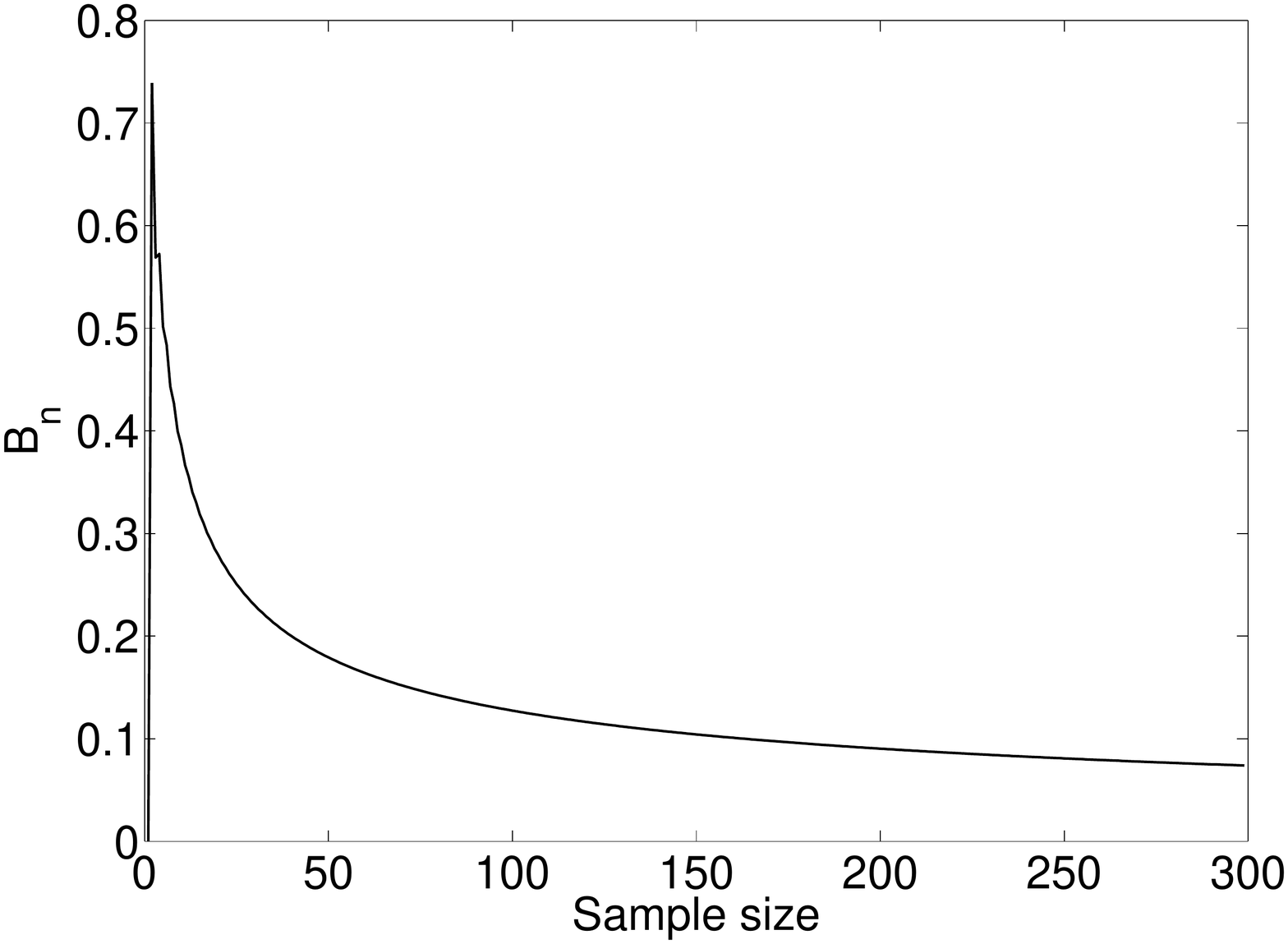}}
\caption{$I_n$ and $B_n$ with respect to the sample size $1\le n \le 300$ when inputs are deterministic $(8/5)(i-1)/(n-1)$, $i=1,\dots , n$,  the transfer function is $h(x)=1-(x-4/5)^2$ on $[0,8/5]$, and the outputs are, or are not, compromised by Gaussian intrusions $\varepsilon_i\sim N(0,\sigma^2_{\varepsilon})$.}
\label{fig-deter}
\end{figure}
we present both compromised and non-compromised cases. In the compromised case (the two top panels), the convergence of $I_n$ to $1/2$ is even faster than in Figure~\ref{fig-contam-confuse}, whereas the sequence $B_n$ grows in a similar fashion as in Figure~\ref{fig-contam-confuse}. Hence, using the deterministic inputs, we reach the same conclusion as before (i.e., the presence of intrusion variables) but in a much faster fashion.

In the non-contaminated case (the two bottom panels of Figure~\ref{fig-deter}), $I_n$ not just converges to $1/2$ but is actually equal to $1/2$, which may give the impression about the presence of intrusions due to $h(a)=h(b)$. However, the declining and thus asymptotically bounded sequence $B_n$ removes all the doubts by correctly implying the absence of intrusions (Case 2(ii) of the rule of thumb). Hence, paying particular attention to the asymptotic behaviour of $B_n$ is indeed enlightening. As to $I_n$ and related interpretations, we refer to Chen at al.~(2018), and Gribkova and Zitikis~(2018).

In Section~\ref{rule} we presented three theorems that had justified the rule of thumb when inputs were random, but the theorems, naturally, did not cover the deterministic case. The following three theorems, which mimic though are not identical to  Theorems~\ref{Thm-A}--\ref{Thm-C}, show that the rule of thumb works in the case of deterministic inputs as well.

\begin{theorem}\label{Thm-Adet}
The deterministic inputs given by equation (\ref{eq-1c}) produce outputs in reasonable order, as per Definition~\ref{order}, that is, $B_n^0$ with $x_{i,n}$ instead of $X_{i:n}$ is asymptotically bounded. Hence, if the sequence $B_n$ is also asymptotically bounded and the intrusion variables are iid $\varepsilon_{1},\dots , \varepsilon_{n}\sim F_{\varepsilon}$ with means zero, then the cdf $F_{\varepsilon} $ must be degenerate (at point~$0$).
\end{theorem}

\begin{proof}
Since the derivative $h'(x)$ exists and is uniformly bounded on $[a,b]$, we have
\begin{align}
B_n^0&={1\over \sqrt{n}} \sum_{i=2}^n \big | h(x_{i,n})-h(x_{i-1,n}) \big |
\notag
\\
&\le {c\over \sqrt{n}} \sum_{i=2}^n \big | x_{i,n}-x_{i-1,n} \big |
\notag \\
&=O_{\mathbf{P}}(n^{-1/2}).
\label{bound-b0}
\end{align}
Hence, just like in the proof of Theorem~\ref{Thm-A}, we conclude that $B_n\ge \sqrt{n}~ \mathbf{E}\big [|\varepsilon_{2}-\varepsilon_{1}|\big ]+ O_{\mathbf{P}}(1)$, which in turn implies that the cdf $F_{\varepsilon} $ is degenerate.
\end{proof}

Reflecting upon the proof of Theorem~\ref{Thm-Adet}, we actually need only $B_n^0= O_{\mathbf{P}}(1)$ and can thus relax the uniform boundedness of $h'(x)$. For example, we can require the transfer function $h(x)$ to be $\gamma$-H\"{o}lder continuous for some $\gamma \ge 1/2$, which means that there is a constant $c<\infty $ such that  $|h(x)-h(y)|\le c|x-y|^{\gamma }$ for all $x,y\in [a,b]$. As an illustration, the transfer function $h(x)$ with a uniformly on $[a,b]$ bounded derivative corresponds to the case $\gamma=1$. For the sake of clarity, however, we continue working with $h(x)$ whose derivative is uniformly bounded, which is a practically attractive and justifiable assumption.

\begin{theorem}
\label{Thm-Bdet}
Let the inputs be deterministic and given by equation (\ref{eq-1c}), and let the intrusion variables $\varepsilon_{1},\dots , \varepsilon_{n}$ follow a non-degenerate cdf $F_{\varepsilon} $. Then $I_n$ converges to $ 1/2$.
\end{theorem}

\begin{proof}
Given the facts noted in the first half of the proof of Theorem~\ref{Thm-A}, and also using bound (\ref{bound-b0}), we have
\begin{align*}
B_n&= {1\over \sqrt{n}} \sum_{i=2}^n | \varepsilon_{[i]}-\varepsilon_{[i-1]} |+O(1)
\\
&= {1\over \sqrt{n}} \sum_{i=2}^n \Big (| \varepsilon_{[i]}-\varepsilon_{[i-1]} |
- \mathbf{E}\big [| \varepsilon_{[i]}-\varepsilon_{[i-1]} |\big ]\Big )
+ {n-1\over \sqrt{n}} \mathbf{E}\big [| \varepsilon_{[i]}-\varepsilon_{[i-1]} |\big ]
+ O(1)
\\
&=\sqrt{n}~ \mathbf{E}\big [|\varepsilon_{2}-\varepsilon_{1}|\big ]+ O_{\mathbf{P}}(1).
\end{align*}
Following analogous arguments, we prove that, when $n\to \infty $,  
\[
A_n:={1\over \sqrt{n}} \sum_{i=2}^n \big ( Y_{i,n}-Y_{i-1,n} \big )_{+}
=\sqrt{n}~ \mathbf{E}\big [(\varepsilon_{2}-\varepsilon_{1})_{+}\big ]+ O_{\mathbf{P}}(1). 
\]
Since $I_n=A_n/B_n$, we have
\begin{equation}
I_n={\mathbf{E}\big [ (\varepsilon_{2}-\varepsilon_{1} )_{+}\big ]+ O_{\mathbf{P}}(n^{-1/2})
\over
\mathbf{E}\big [|\varepsilon_{2}-\varepsilon_{1}|\big ]+ O_{\mathbf{P}}(n^{-1/2})}
\stackrel{\mathbf{P}}{\longrightarrow}
{\mathbf{E}\big [ (\varepsilon_{2}-\varepsilon_{1} )_{+}\big ]
\over \mathbf{E}\big [|\varepsilon_{2}-\varepsilon_{1}|\big ]}
\label{eq-30}
\end{equation} 
when $n\to \infty $. The ratio on the right-hand side of statement (\ref{eq-30}) is equal to $1/2$ because the random variable $\varepsilon_{2}-\varepsilon_{1}$ is symmetric  irrespective of the cdf $F_{\varepsilon} $. This proves Theorem~\ref{Thm-Bdet}. 
\end{proof}

\begin{theorem}[Chen et al., 2018]
\label{Thm-Cdet}
Let the inputs be deterministic and given by equation (\ref{eq-1c}). Furthermore, assume that the outputs are not compromised by intrusion variables, that is, the cdf $F_{\varepsilon} $ is degenerate. If the derivative $h'(x)$ is $\gamma$-H\"{o}lder continuous for some $\gamma >0$, then $I_n$ converges to
\[
I(h):={\int_{a}^{b} (h'(u))_{+}\mathrm{d}u \over \int_{a}^{b}|h(u)|\mathrm{d}u}.
\]
\end{theorem}

Hence, in view of Theorems~\ref{Thm-Adet}--\ref{Thm-Cdet}, we conclude that if the control system is in reasonable order as per Definition~\ref{order}, then even when the inputs are deterministic, the rule of thumb can distinguish between compromised and non-compromised outputs. To illustrate, in Figure~\ref{fig-noise-deter}
\begin{figure}[t!]
\centering
\subfigure[$I_n$ when $X_i\sim \mathrm{Unif}(0,1)$.]{%
\includegraphics[height=0.3\textwidth]{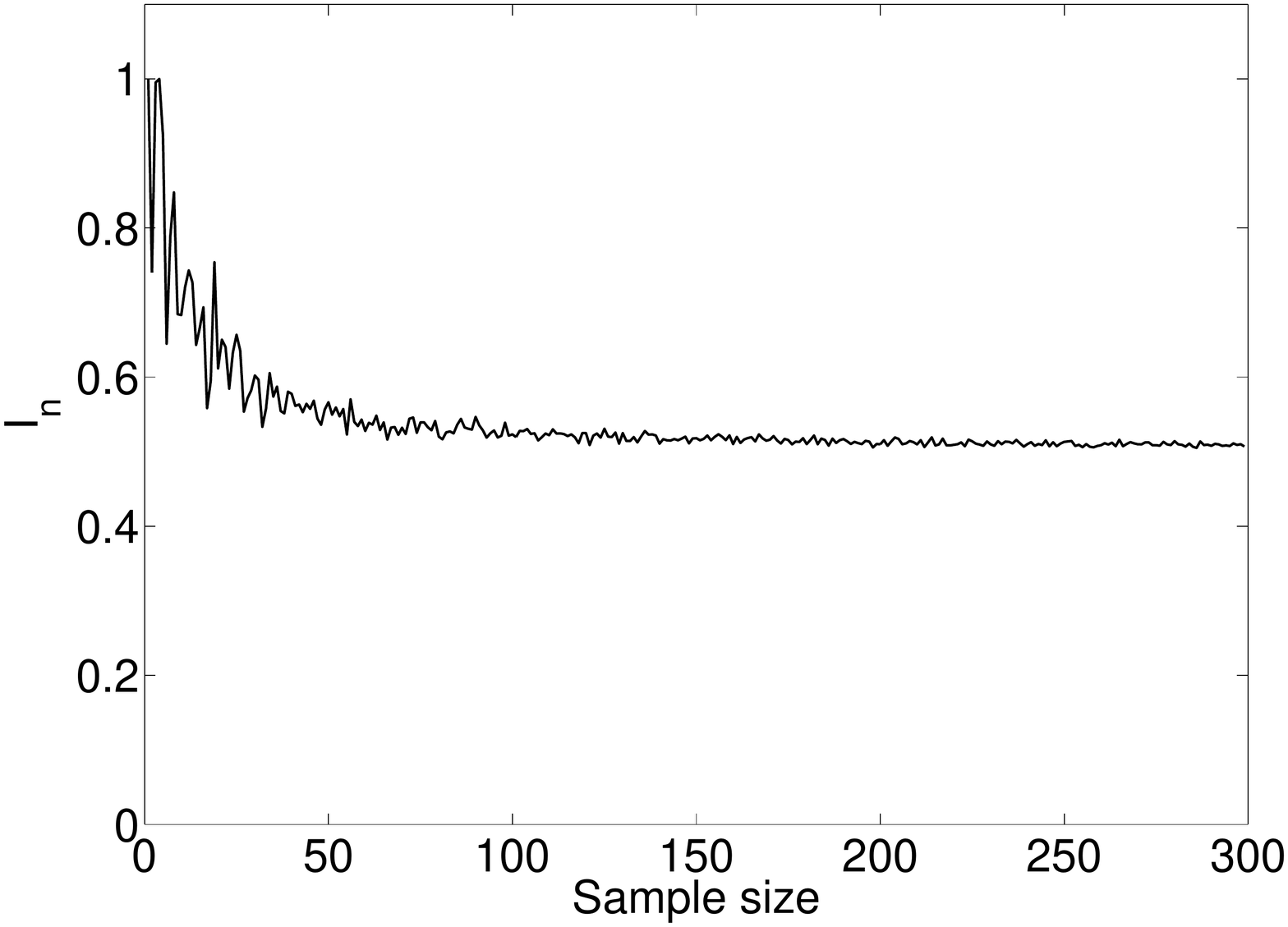}}
\hfill
\subfigure[$B_n$ when $X_i\sim \mathrm{Unif}(0,1)$.]{%
\includegraphics[height=0.3\textwidth]{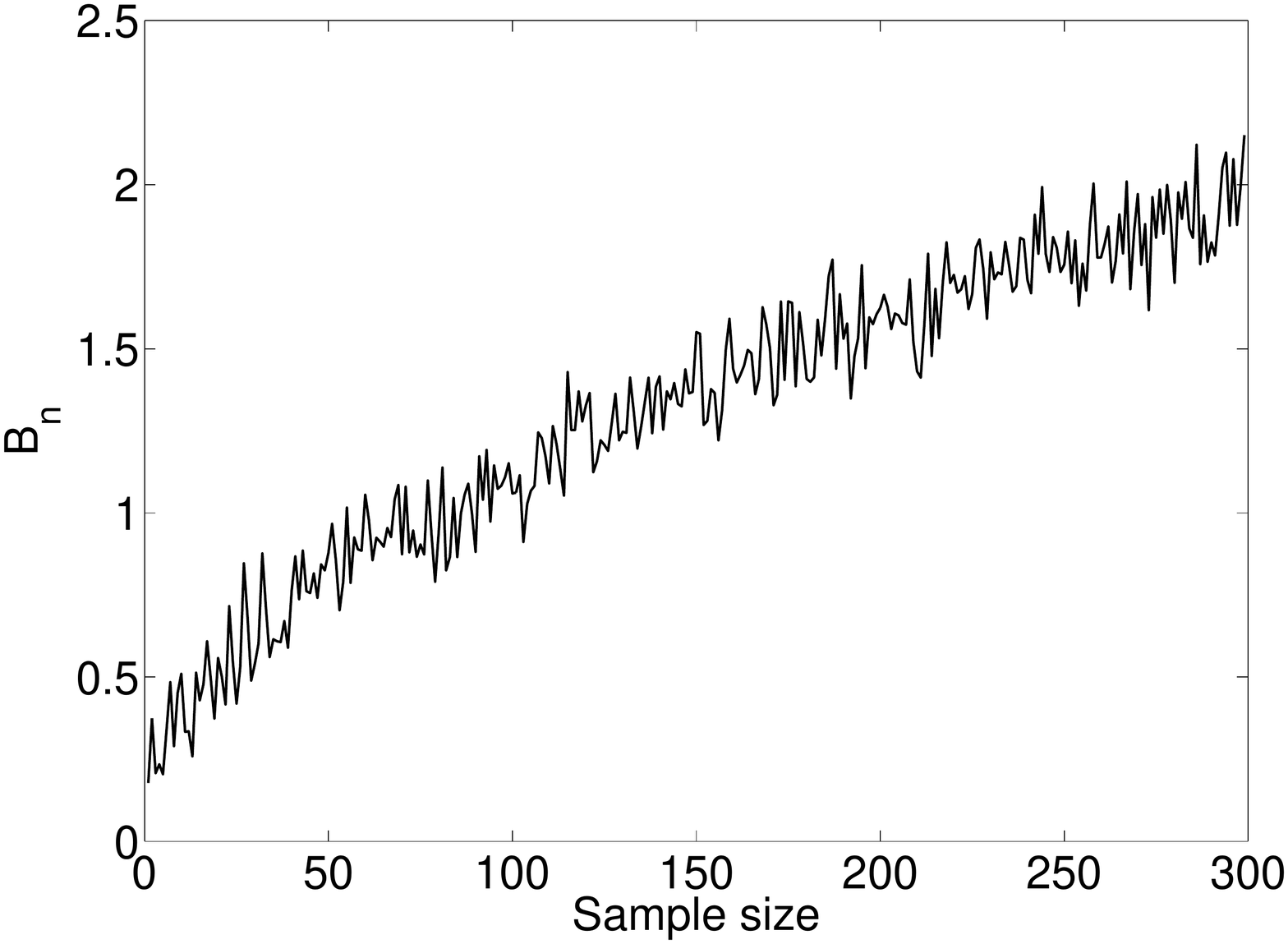}}
\\
\subfigure[$I_n$ when inputs are ${i-1 \over n-1}$, $i=1,\dots , n$.]{%
\includegraphics[height=0.3\textwidth]{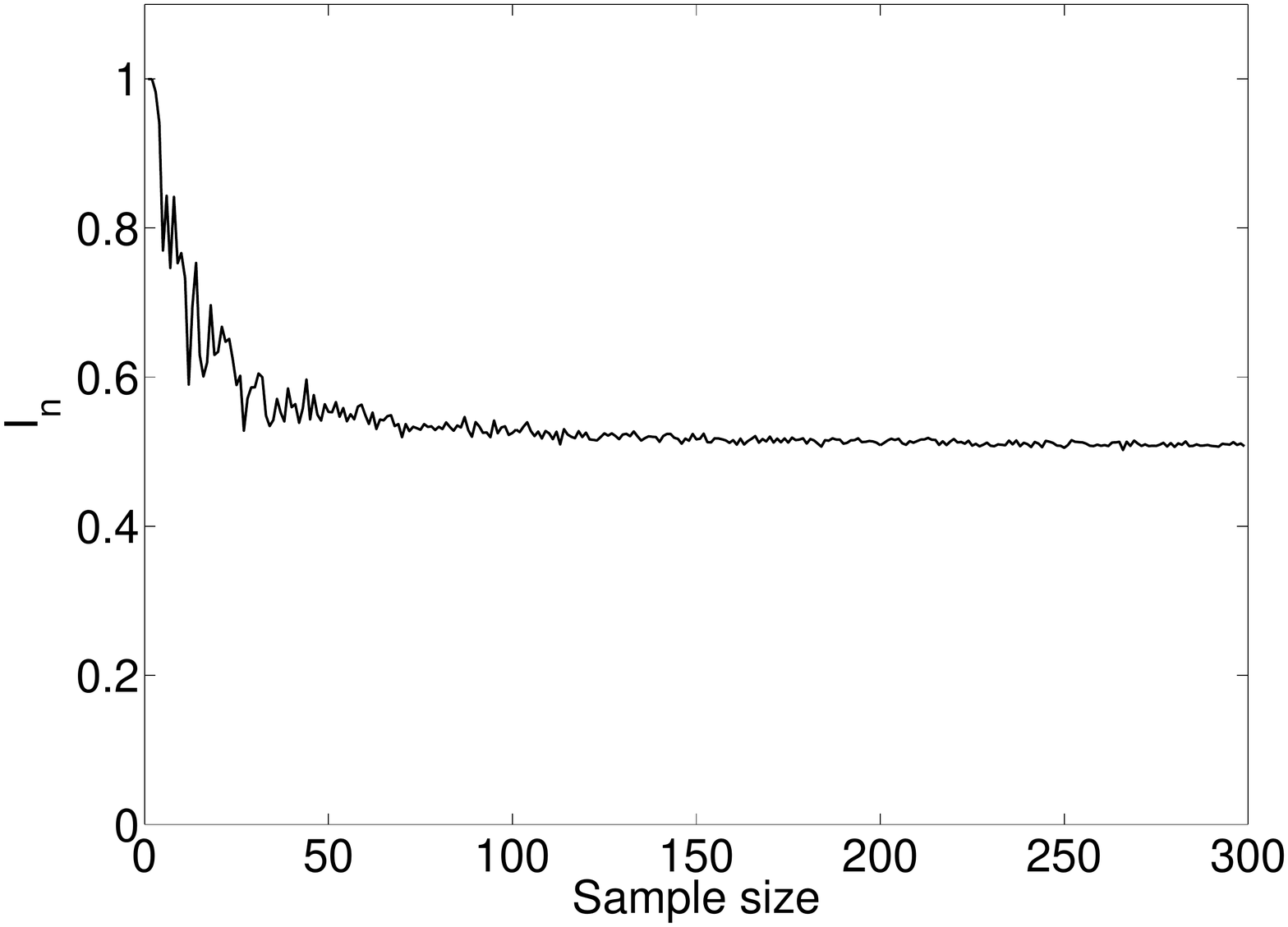}}
\hfill
\subfigure[$B_n$ when inputs are ${i-1 \over n-1}$, $i=1,\dots , n$.]{%
\includegraphics[height=0.3\textwidth]{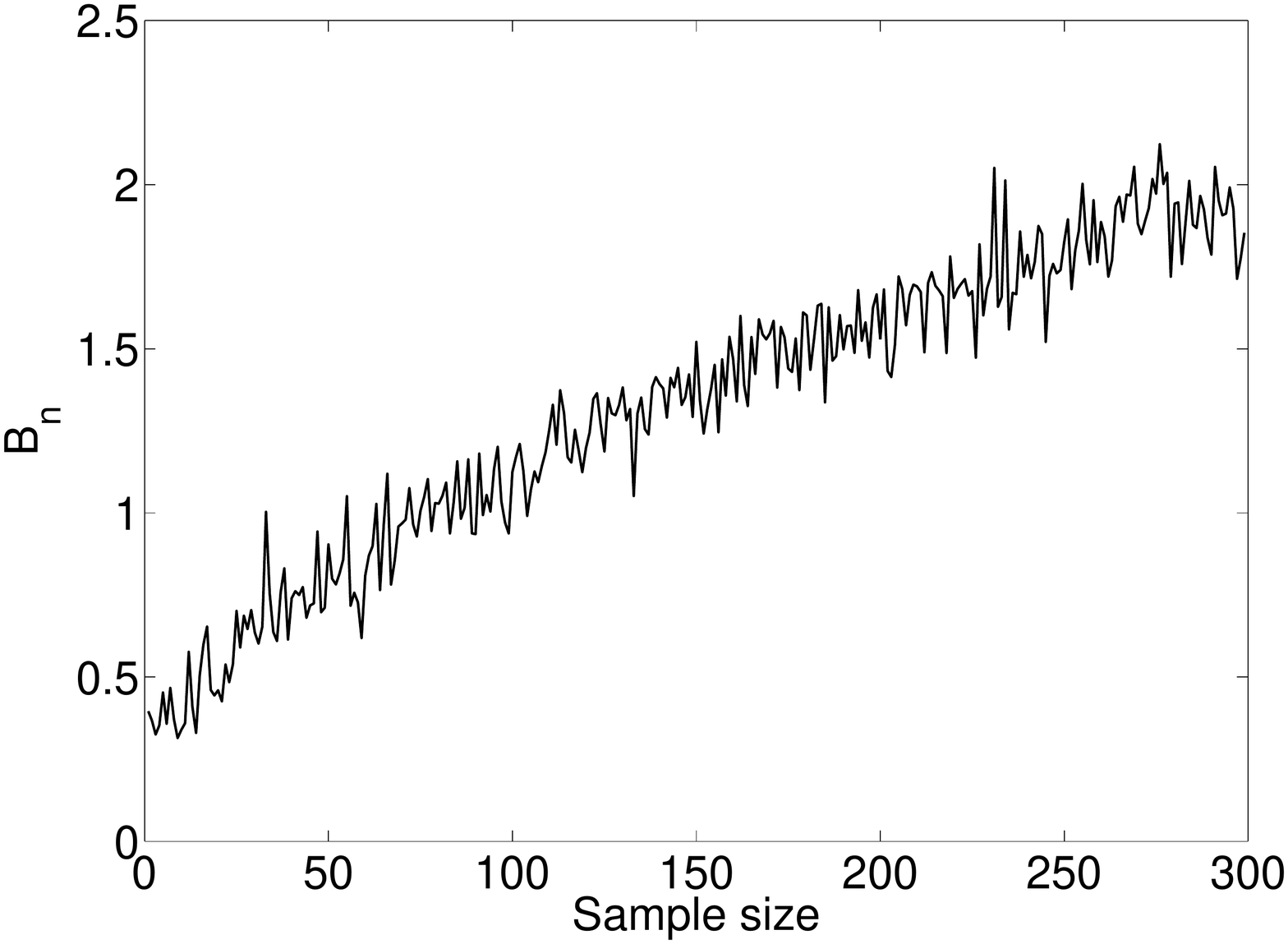}}
\caption{$I_n$ and $B_n$ with respect to the sample size $1\le n \le 300$ when $h(x)=1-(x-4/5)^2$ on $[0,1]$, and compromised outputs with Gaussian intrusions  $\varepsilon_i\sim N(0,\sigma^2_{\varepsilon})$ with $\sigma^2_{\varepsilon}=0.01$.}
\label{fig-noise-deter}
\end{figure}
we depict the asymptotic behaviour of $I_n$ and $B_n$ in the case of random and deterministic inputs, and when the outputs are being compromised by intrusions $\varepsilon_{1},\dots , \varepsilon_{n}$, which for illustrative purposes are assumed to be iid Gaussian with the same strictly-positive variances. The two types of inputs, though different in nature, are uniform in their respective ways: in the random case, they give rise to the uniform on $[0,1]$ order statistics, whereas in the deterministic case, the points $(i-1)/(n-1)$, $i=1,\dots , n$, split the interval $[0,1]$ into perfectly equal subintervals of length $1/(n-1)$. The asymptotic behaviours of $I_n$ and $B_n$ in the two cases are more or less identical, unlike what we shall see next, when the system is not compromised.

Figure~\ref{fig-pure-deter}
\begin{figure}[t!]
\centering
\subfigure[$I_n$ when $X_i\sim \mathrm{Unif}(0,1)$.]{%
\includegraphics[height=0.3\textwidth]{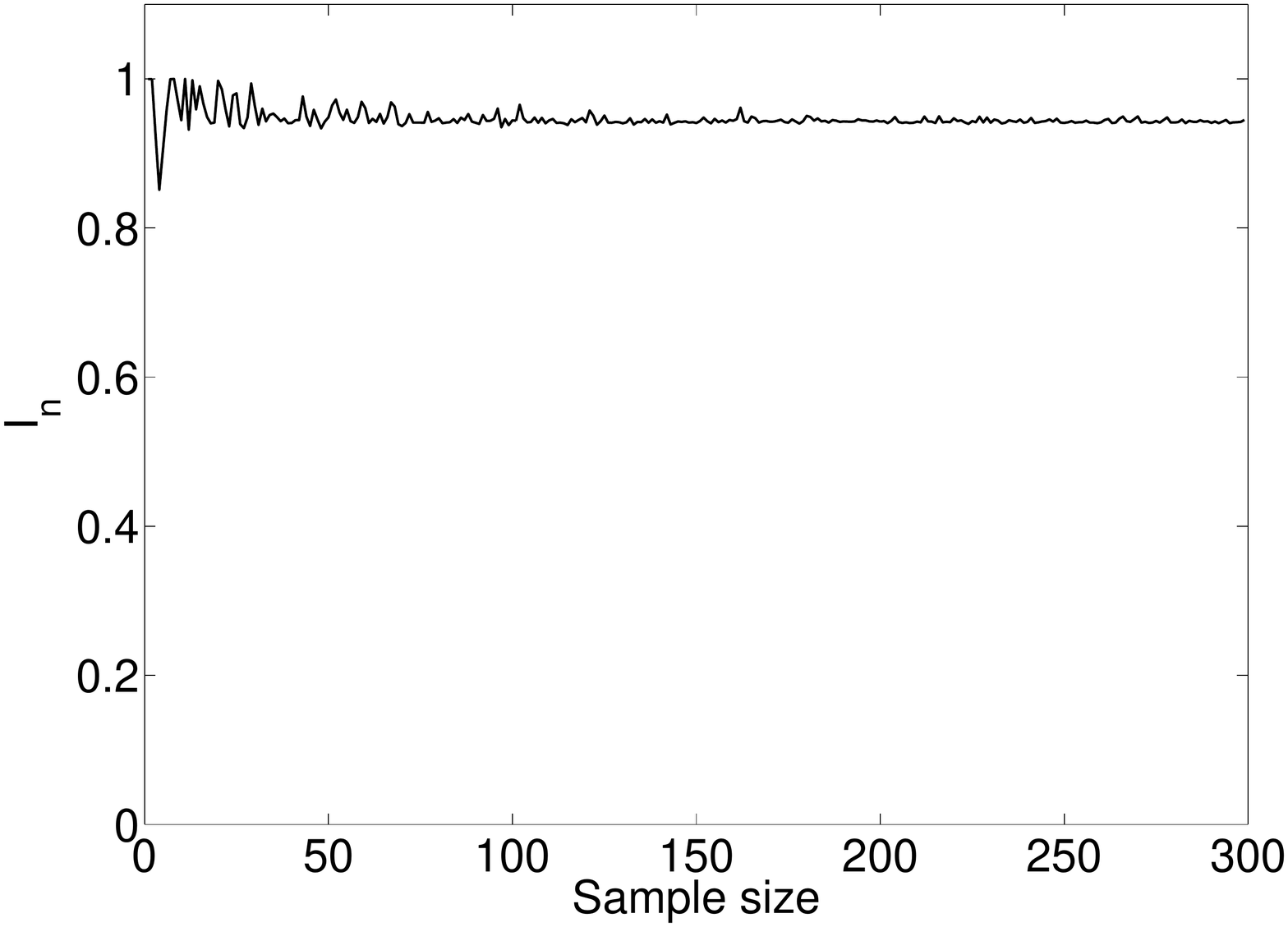}}
\hfill
\subfigure[$B_n$ when $X_i\sim \mathrm{Unif}(0,1)$.]{%
\includegraphics[height=0.3\textwidth]{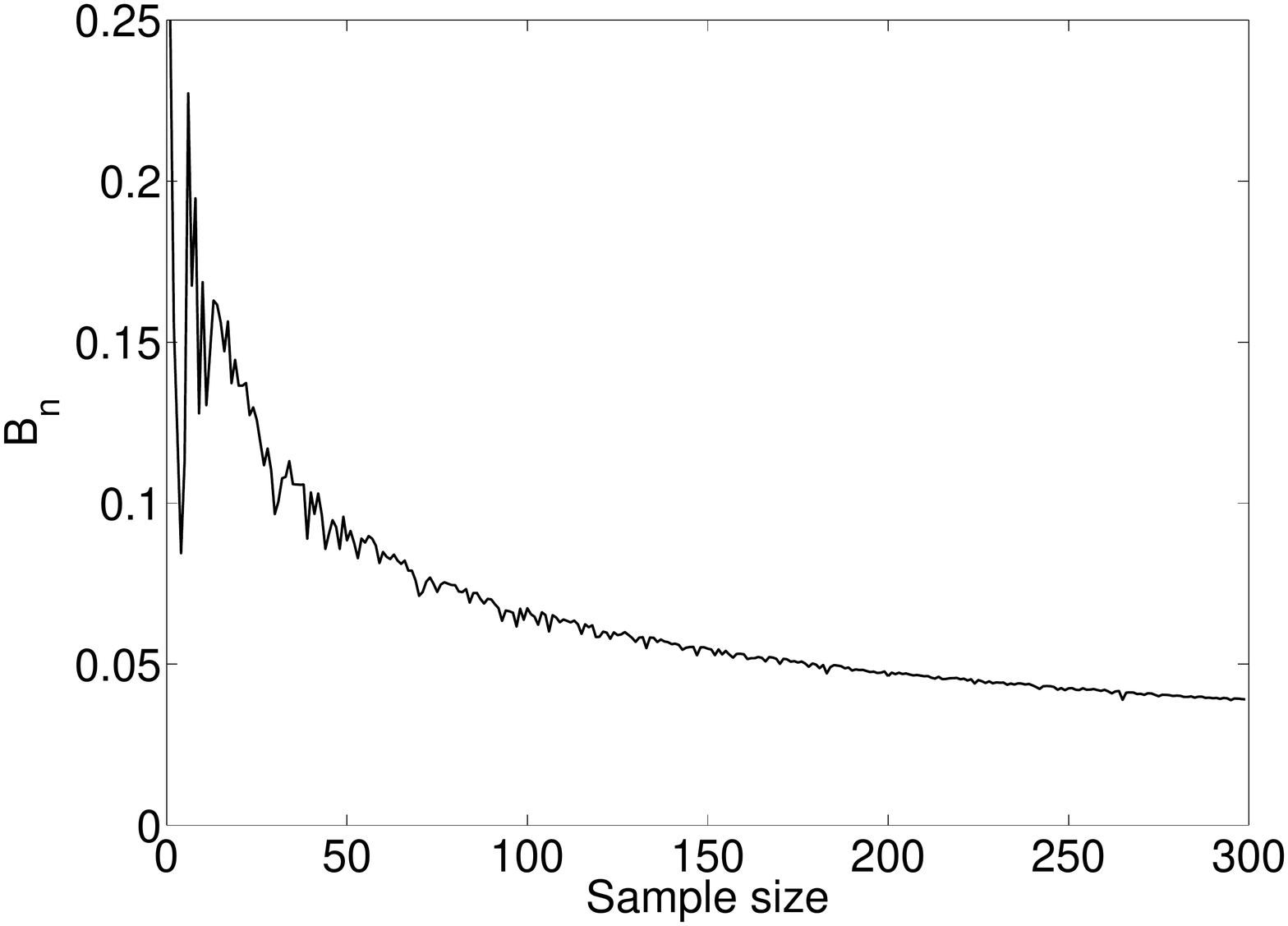}}
\\
\subfigure[$I_n$ when inputs are $x_{i,n}$.]{%
\includegraphics[height=0.3\textwidth]{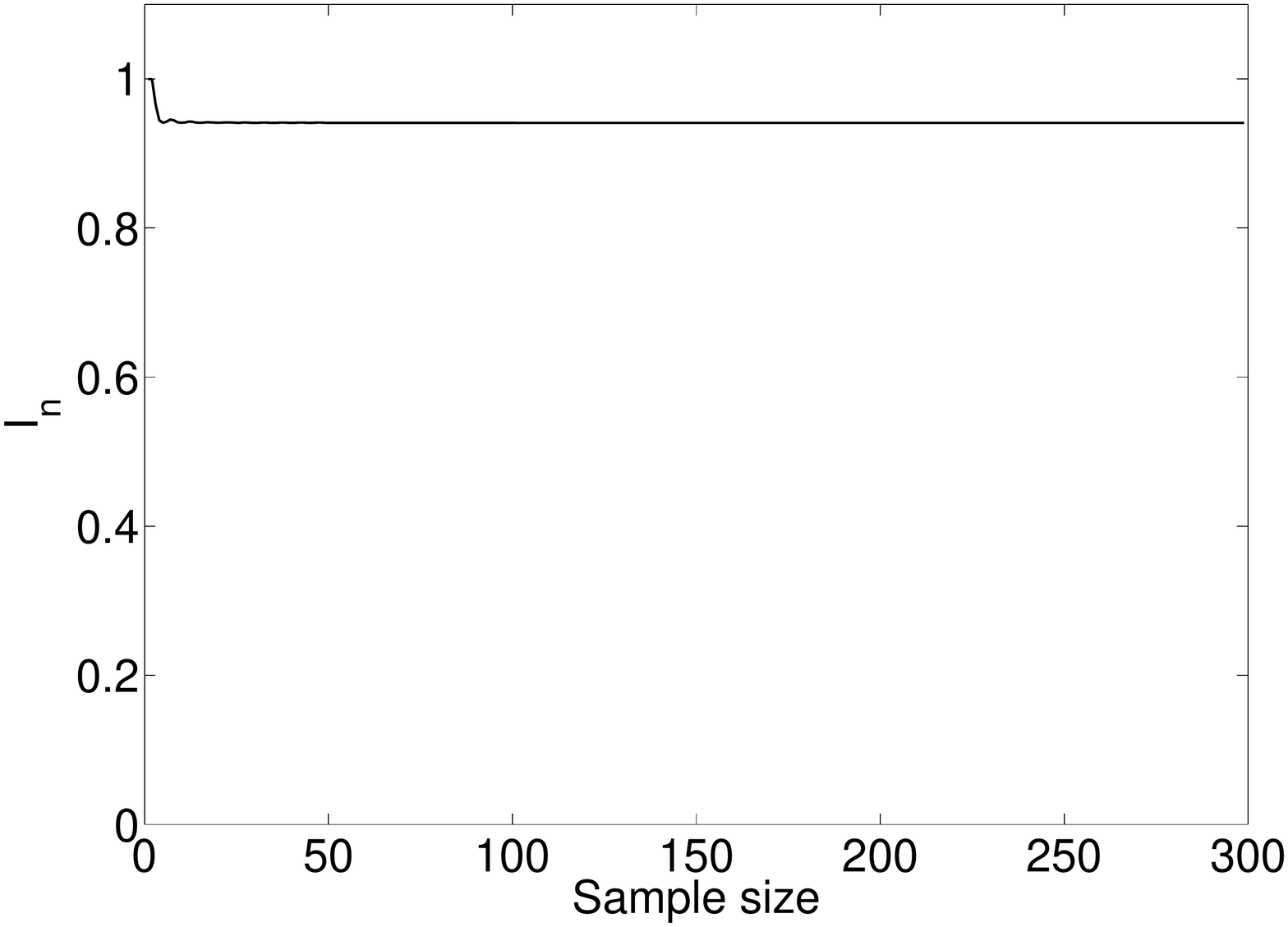}}
\hfill
\subfigure[$B_n$ when inputs are $x_{i,n}$.]{%
\includegraphics[height=0.3\textwidth]{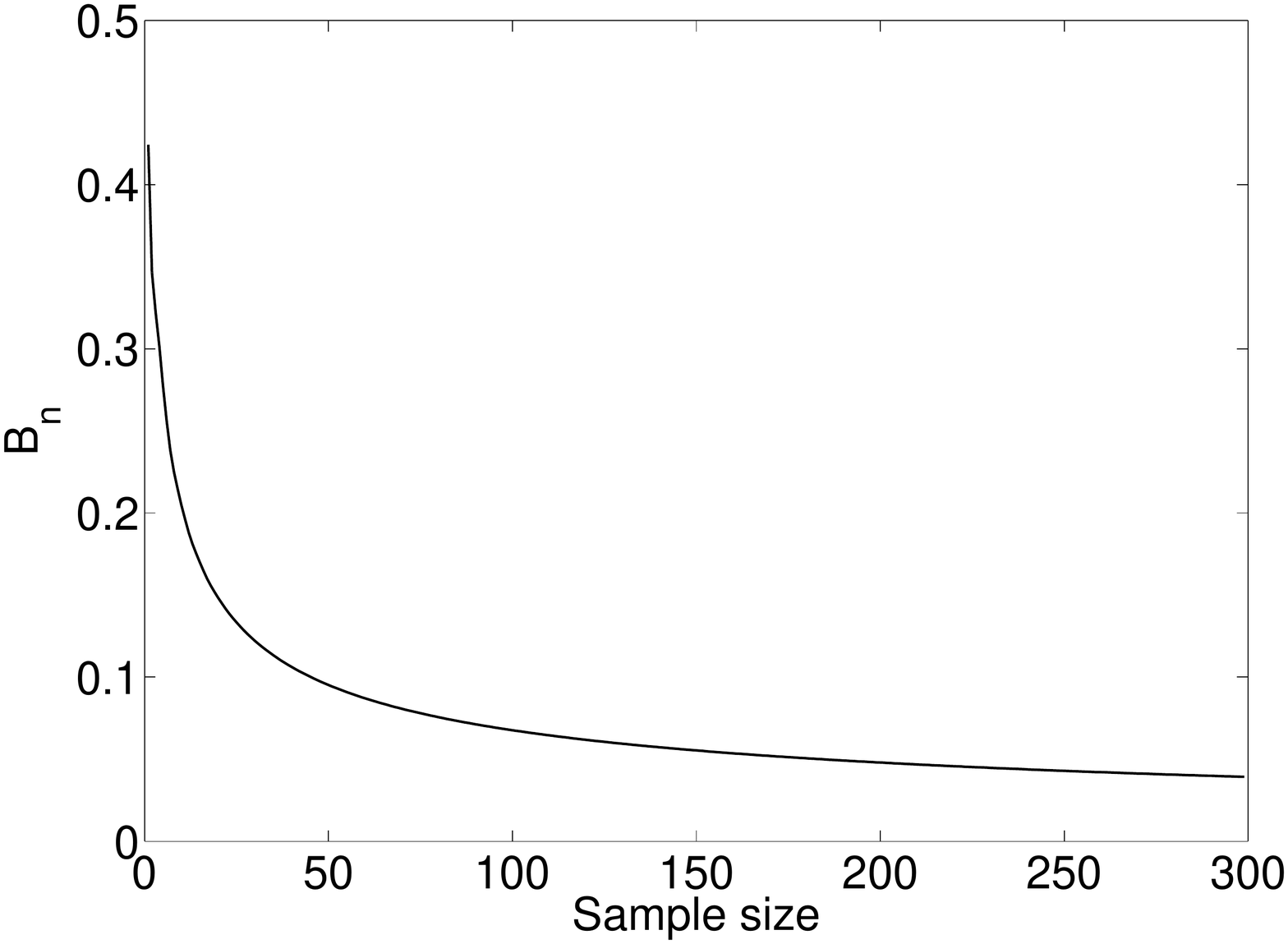}}
\caption{$I_n$ and $B_n$ with respect to the sample size $1\le n \le 300$ when $h(x)=1-(x-4/5)^2$ on $[0,8/5]$, and non-compromised outputs.}
\label{fig-pure-deter}
\end{figure}
is based on the same numerical example but without intrusion variables added to the outputs. Using the rule of thumb we confidently conclude that the system is not compromised. Indeed, the sequence $I_n$ rapidly converges to a limit different from $1/2$, and $B_n$ is clearly  asymptotically bounded. It is interesting to compare the cases of randomly (the two top panels of Figure~\ref{fig-pure-deter}) and deterministically (the two bottom panels) uniform inputs. In particular, using even very small sample sizes, the deterministic case allows us to very quickly conclude the absence of intrusion variables in the control system.

\section{Concluding notes}
\label{conclude}

We have proposed, justified, and numerically illustrated a rule of thumb for deciding whether or not a control system is being comprised. The rule is easy to implement, and we have discussed its performance under scenarios with random and deterministic inputs. The latter ones are particularly useful for speedy system's performance and vulnerability checks. The rule of thumb has been supported by rigorous theoretical considerations, which not only make up a solid foundation for the rule but also gives rise to the possibility for extending its use beyond what we have described in the present paper.

\section*{Acknowledgement}

Research of the second author has been supported by the Natural Sciences and Engineering Research Council of Canada.

\end{document}